\documentclass[12pt,a4paper,reqno]{amsart}
\usepackage{graphicx} 
\usepackage{natbib}
\usepackage{mathtools}
\usepackage{amsmath}
\usepackage{amsthm}
\usepackage{thmtools}
\usepackage{mathrsfs}
\usepackage{pdflscape}
\usepackage{appendix}
\usepackage{blindtext}
\usepackage{dsfont}
\usepackage[bookmarksnumbered,linktocpage,hypertexnames=false,colorlinks=true,linkcolor=blue,urlcolor=blue,citecolor=blue,anchorcolor=green,breaklinks=true,pdfusetitle]{hyperref}
\usepackage{mathtools}
\usepackage[capitalize]{cleveref}
\usepackage{array}
\usepackage{enumitem}
\usepackage{tikz-cd}
\usepackage{comment}
\usetikzlibrary{calc,decorations.pathmorphing,shapes}

\DeclareMathOperator{\id}{id}
\DeclareMathOperator{\spa}{span}

\DeclareMathOperator{\Tr}{Tr}

\DeclareMathOperator{\Fix}{Fix}

\newcommand{\N}{\mathbb{N}}

\newcommand{\R}{\mathbb{R}}
\newcommand{\C}{\mathbb{C}}

\newcommand{\cA}{\mathcal{A}}
\newcommand{\cB}{\mathcal{B}}

\newcommand{\cH}{\mathcal{H}}

\newcommand{\cJ}{\mathcal{J}}
\newcommand{\cK}{\mathcal{K}}

\newcommand{\cR}{\mathcal{R}}

\newcommand{\1}{{\mathds{1}}}

\newcommand{\msA}{\mathscr{A}}

\newcommand{\bra}[1]{\mathinner{\langle #1|}}
\newcommand{\ket}[1]{\mathinner{|#1\rangle}}
\newcommand{\braket}[2]{\mathinner{\langle #1|#2\rangle}}
\newcommand{\dyad}[2]{| #1\rangle \langle #2|}

\newcommand{\bit}{\{0,1\}}
\allowdisplaybreaks[4]

\newtheorem*{quesnocounter}{Question}
\newtheorem*{theoremnocounter}{Theorem}
\newtheorem{thm}{Theorem}[section]
\newtheorem{lem}[thm]{Lemma}

\newtheorem{prop}[thm]{Proposition}

\theoremstyle{definition}
\newtheorem{defn}[thm]{Definition}
\newtheorem{exa}[thm]{Example}
\theoremstyle{remark}

\DeclarePairedDelimiter\abs{\lvert}{\rvert}
\DeclarePairedDelimiter\norm{\lVert}{\rVert}

\DeclarePairedDelimiter\parens{\lparen}{\rparen}

\DeclarePairedDelimiter\bracks{\lbrack}{\rbrack}

\newcommand{\op}{\oplus}

\newcommand{\snote}[1]{{\textcolor{purple}{[Sigurd: #1]}}}

\definecolor{darkgreen}{rgb}{0,0.5,0}
\newcommand{\YZ}[1]{{\textcolor{darkgreen}{[Yuming: #1]}}}

\title[Guess your neighbor's input: Quantum advantage in Feige's game]{Guess your neighbor's input:\texorpdfstring{\\}{} Quantum advantage in Feige's game}

\author{Simon Schmidt, Sigurd A.~L.~Storgaard, Michael Walter, Yuming Zhao}

\address{Simon Schmidt, Faculty of Computer Science, Ruhr University Bochum, Germany}
\email{s.schmidt@rub.de}
\address{Sigurd Storgaard, Department of Mathematical Sciences, University of Copenhagen, Denmark}
\email{sals@math.ku.dk}
\address{Michael Walter, Ludwig-Maximilians-Universit\"at M\"unchen, Germany; Faculty of Computer Science, Ruhr University Bochum, Germany; Korteweg-de Vries Institute for Mathematics, University of Amsterdam, Netherlands}
\email{michael.walter@lmu.de}
\address{Yuming Zhao, Department of Mathematical Sciences, University of Copenhagen, Denmark}
\email{yuming@math.ku.dk}

\begin{document}

\begin{abstract}
In this article, we study a nonlocal game with two questions and three answers per player, which was first considered by Feige in 1991, and show that there is quantum advantage in this game. We prove that the game is a robust self-test for the $3$-dimensional maximally entangled state.
Furthermore, we show that the game can be seen as the "or" of two games that each do not have quantum advantage. Lastly, we investigate the behavior of the game with respect to parallel repetition in the classical, quantum and non-signalling case and obtain perfect parallel repetition of the non-signalling value if Feige's game is repeated an even amount of times.
\end{abstract}

\maketitle

\tableofcontents

\section{Introduction}
A \emph{nonlocal game} \cite{CHTW} consists of two (or more) cooperative players, often called Alice and Bob, that interact with a referee. In the game, the referee sends a question to each player, for which each player gives  an answer. Based on the tuple of questions and answers, the referee decides if the players win or lose. Communication is not permitted between Alice and Bob, hence each player has no information about the questions given to the other players, nor do they know the answers provided to the referee by the other players. Nevertheless, the description of the game is known to the players ahead of time, allowing them to agree on a strategy beforehand and maximize their probability of winning the game. The \emph{classical value} $\omega_c(G)$ of a nonlocal game $G$ is the maximum winning probability of classical players, and the \emph{quantum value} $\omega_q(G)$ denotes the maximum winning probability of quantum players sharing a (finite) amount of quantum resources (such as entangled quantum states, like EPR pairs). Furthermore, the \emph{non-signalling value} $\omega_{ns}(G)$ is the maximum winning probability if one considers probability distributions such that the marginal distribution of either player’s answers must be independent of the other player’s question, without requiring the probability distribution to be physically realizable.

\noindent\textbf{\emph{Feige's game.}} In this work, we revisit a game first studied by Feige \cite{Feige}, where it was considered as counterexample in parallel repetition since its value does not decrease when played twice in parallel. The game is defined as follows. Alice and Bob both receive a bit as question (with uniform probability), and they are each allowed to choose an answer from the set $\{0,1,\perp\}$. To win the game, exactly one of the players has to answer with $\perp$, while the other player's answer has to agree with the question the player who answered $\perp$ received. More formally, the game $G_F$ \cite{Feige} is defined via
$$
    X= Y=\{0,1 \}, \quad A=B=\{0,1,\perp \}, \quad \pi(x,y)=\tfrac{1}{4} \quad  \text{for all } (x,y)
$$
and has the predicate
$$
    V(a,b,x,y)=\begin{cases}
        1 \quad (a,b)=(\perp,x) \text{ or } (a,b)=(y,\perp),
        \\
        0 \quad \text{otherwise.}
    \end{cases}
$$
It has been conjectured in the literature that $\omega_q(G_F)=\omega_c(G_F)$ (\cite{yuen2016phd}), but so far no attempts for computing the quantum value of the game were made. Since there is no quantum advantage for other variants of "guess your neighbors input" games (\cite{almeidaguess}, \cite{acinguess}), it sounds plausible that there is also no quantum advantage in Feige's game. We show that, surprisingly, there exists a quantum strategy for Feige's game that exceeds the classical value and we compute the quantum value.

\begin{theoremnocounter}
Feige's game $G_F$ has quantum advantage with $\omega_q(G_F)=\frac{9}{16}>\frac{1}{2}=\omega_c(G_F)$.
\end{theoremnocounter}

Interestingly, the players use operators acting on $\C^3$ in the optimal quantum strategy and they share the maximally entangled state $\ket{\psi_3}=\sum_{i=0}^2 \ket{i}\otimes \ket{i}$. The strategy that achieves a winning probability of $\frac{9}{16}$ is constructed from the operators
\begin{align*}
    \tilde Z = Z \op 1 \quad\text{and}\quad \tilde X = X \op 1 \quad\text{on}\quad \mathbb{C}^3 = \mathbb{C}^2 \op \mathbb{C},
\end{align*}
where $X$ and $Z$ denote the usual Pauli matrices. The projective measurements are given by
\begin{align} \label{proj_meas}
    &P^x_a= \tilde X {\tilde Z}^x \tilde X \ket{\tilde a}\bra{\tilde a} \tilde X {\tilde Z}^x \tilde X, \quad a \in \{0,1,\bot\} \quad \text{ and}
    &&Q^y_b= \tilde X P^y_b \tilde X, \quad b \in \{0,1,\bot\},
\end{align}
for a carefully chosen basis $\{\ket{\tilde a}\}$ of $\C^3$. To obtain this basis, we crucially use that $\tilde Z$ and $\tilde X$ are two non-commuting observables such that $\tilde Z$ and $\tilde X \tilde Z \tilde X$ commute.

An approach to upper-bounding the quantum value of a nonlocal game is to associate a formal game polynomial $p$ to a game $G$, whose variables are identified with abstract projective measurements $P^{x}_a$ and $Q^{y}_b$ such that $\bra{\psi}p(P^{x}_a,Q^y_b)\ket{\psi}$ is precisely the winning probability of a quantum (commuting) strategy for $G$. To prove that the quantum value of $G_F$ is at most $\frac{9}{16}$, we find a \emph{sums of square decomposition} of the game polynomial of Feige's game by using the computer. More precisely, we obtain a vector of formal variables $V=V(P^{x}_a,Q^y_b)$ and a positive real matrix $Y$ such that $V^*YV=\frac{9}{16}I-p(P^{x}_a,Q^y_b)$, which implies $\omega_q(G_F)\leq \frac{9}{16}$.

\noindent\textbf{\emph{Self-testing.}} A nonlocal game is a \emph{self-test} if its optimal quantum strategy is unique up to local isometries, as introduced in \cite{MY}. Self-testing is a useful property of a nonlocal game as it can be used for certifying quantum devices. We have the following theorem for Feige's game.

\begin{theoremnocounter}
The game $G_F$ is a robust self-test for the operators in (\ref{proj_meas}) and the maximally entangled state $\ket{\psi_3}=\sum_{i=0}^2 \ket{i}\otimes \ket{i}$.
\end{theoremnocounter}

While self-tests for all maximally entangled states are known (\cite{MPS, fu2022constant}), we believe that this is the smallest game that self-tests $\ket{\psi_3}$ with just two questions and three answers per player.

In fact, we obtain the result by showing that a certain interesting algebra has a unique irreducible representation. We deduce the algebra from the sums of square decomposition we have seen before. Since it holds $V^*YV=0$ in an optimal strategy, there are relations that the operators in the strategy have to fulfill. In our case we get the relations for Alice's operators
\begin{align*}
P^x_a P^{1-x}_{a'} P^x_a\ket{\psi}  &= \nu_{aa'} P^x_a\ket{\psi} \text{ for all }x\in \{0,1 \}, a,a'\in \{ 0,1,\bot\}, \text{ where }
\nu = \begin{bmatrix}
 \frac{9}{16} & \frac{1}{16} & \frac{3}{8} \\
 \frac{1}{16} & \frac{9}{16} & \frac{3}{8}  \\
 \frac{3}{8} & \frac{3}{8} & \frac{1}{4}
\end{bmatrix}.
\end{align*}
The same relations hold for Bob's operators. Note that the relations look similar to the relations obtained from mutually unbiased measurements \cite{MUM_Tavakoli_etal}, where for the latter it holds $\nu=\frac{1}{d}J$ (where $J$ is the matrix of all ones). In case of $P^x_a=\ket{\phi^x_a}\bra{\phi_a^x}$ being rank $1$ projections, the transition matrix from the basis $\{\ket{\phi^x_a}\}$ to the basis $\{\ket{\phi^{1-x}_a}\}$ is given by a unitary matrix $H$ with $|h_{aa'}|=\sqrt{\nu_{aa'}}$ (which is a Hadamard matrix in the case of $\nu=\frac{1}{d}J$). Similar to some Hadamard matrices being unique up to equivalence \cite{Haagerup}, we show that the same is true for unitaries with $|h_{aa'}|=\sqrt{\nu_{aa'}}$ for the specific choice of $\nu$ from above. More generally, we give simple criterion for when such a unitary is unique up to a certain kind of equivalence. Using this, and the operator-algebraic framework for self-testing developed in \cite{zhao24}, we obtain that Feige's game is a self-test. Another ingredient we need in our proof is that in an optimal strategy, the actions of Bob's projective measurements on the state are determined by actions of (self-adjoint polynomials in) Alice's projective measurements.

\noindent\textbf{\emph{Feige's game as "or"-game.}} Feige's game has another interesting property, it can be seen as the "or" of two games (\cite{mancinskaschmidt}). For two nonlocal games $G_1$ and $G_2$, the $(G_1\lor G_2)$-game is defined as follows. The referee sends Alice and Bob a pair of questions $(x_1, x_2)$ and $(y_1, y_2)$, respectively, where $x_i, y_i$ are questions in $G_i$, $i\in \{1,2\}$. Each of them chooses one of the questions they received and responds with an answer from the corresponding game. To win the game, two conditions have to be fulfilled:
\begin{itemize}
    \item[(1)] Alice and Bob have to give answers from the same game,
    \item[(2)] their answers have to win the corresponding game.
\end{itemize}

Since in Feige's game, one of the players has to guess the input of the other player, it can be realized as the "or"-game of the following two games: In the first game, Alice always guesses the input of Bob, while he outputs $\perp$ all the time and in the second game the roles are reversed. It is intriguing to note that the games for which a fixed player guesses the input of the other player do not have quantum advantage. We therefore have an example of an "or"-game whose quantum value is strictly higher than the maximum of the quantum values of the games it is constructed from.

\begin{theoremnocounter}
There are nonlocal games $G_1$ and $G_2$ such that $G_F=G_1\lor G_2$ and
\begin{align*}
\omega_q(G_F)=\frac{9}{16}>\frac{1}{2}=\mathrm{max}\{\omega_q(G_1),\omega_q(G_2)\}.
\end{align*}
\end{theoremnocounter}

This theorem is in contrast to the case of perfect quantum strategies (\cite{mancinskaschmidt}). Here it is known that if the players give answers for one of the games in a perfect strategy of the "or"-game, then this game has to have a perfect quantum strategy as well.

\noindent\textbf{\emph{Parallel repetition.}} The original motivation of Feige was to give a counterexample in parallel repetition as its value does not decrease when played twice in parallel. In the \emph{$n$-fold parallel repetition} $G^{\times n}$ of a game $G$, the players receive tuples of questions $(x_1, \dots,x_n)$ and $(y_1,\dots, y_n)$ from the original game and they answer with tuples $(a_1, \dots,a_n)$ and $(b_1,\dots, b_n)$. Alice and Bob win in $G^{\times n}$ if the questions and answers $x_i, y_i, a_i, b_i$ win the game $G$ for each $i \in \{1,\dots, n\}$. We study the behaviour of the game with respect to parallel repetition for the classical, quantum and non-signalling value. The results are summarized in \Cref{tab:tabintro}.

\begin{table}[ht]
    \centering
    {\setlength{\extrarowheight}{3pt}
    \begin{tabular}{c||c|c|c|c}
        {} & $n=1$ & $n=2$ & $n=3$ & $n$ even \\[3pt]
        \hline
        $\omega_c$ & $\tfrac{1}{2}$ & $\tfrac{1}{2}$ & $\mathbf{\tfrac{5}{16}}$ & $\tfrac{1}{2^{n/2}}$
        \\[5pt]
        $\omega_q$ & $\mathbf{\tfrac{9}{16}}$ & $\mathbf{\tfrac{1}{2}}$ & ? & $\mathbf{\tfrac{1}{2^{n/2}}}$  \\[5pt]
        $\omega_{ns}$ &  $\mathbf{\tfrac{2}{3}}$ & $\mathbf{\tfrac{1}{2}}$ & $\mathbf{\tfrac{1}{{3}}}$ & $\mathbf{\tfrac{1}{2^{n/2}}}$
        \\
    \end{tabular}}\vspace{0.3cm}
    \caption{Summary of values of the $n$-fold parallel repetition of Feige's game, where we highlight the values that were previously not known.}
    \label{tab:tabintro}
\end{table}
In the table, the classical, quantum and non-signalling values agree in the case of an even number of repetitions, which follows from the following theorem.

\begin{theoremnocounter}
It holds $\omega_{ns}(G_F^{\times n})= \frac{1}{2^{n/2}}$ for $n$ even.
\end{theoremnocounter}

Note that this is not the case for odd numbers. From the previous discussion, we know that the classical value does not agree with the quantum value for the original game, and we show that the non-signalling value is even higher. Furthermore, we get that at least the classical and non-signalling value are different for $n=3$, despite this not being the case for $n=2$. We do not know if there is quantum advantage for $n=3$.

\begin{quesnocounter}
    Does the quantum value exceed the classical value for the $3$-fold parallel repetition of Feige's game?
\end{quesnocounter}

We are not aware of a game with quantum advantage, for which the classical and quantum value agree for the $2$-fold parallel repetition, but then they separate again for $n=3$. As we can see in \Cref{tab:tabintro} this happens for the classical and the non-signalling value of Feige's game. Note that the quantum strategy consisting of the optimal quantum strategy for one game and the optimal classical strategy for the $2$-fold parallel repetition performs worse than the optimal classical strategy for $n=3$, since $\frac{9}{16}\cdot\frac{1}{2}=\frac{9}{32}<\frac{5}{16}$. Therefore, there is no natural candidate of a quantum strategy that yields quantum advantage for the $3$-fold parallel repetition. The $\frac{5}{16}$ optimal classical strategy is also interesting: It performs worse on the first two rounds of the game in comparison to the optimal strategy for the $2$-fold parallel repetition, but is constructed such that if the players win the first two rounds, they win the third round automatically.

\section{Preliminaries}
Let $A,B,X,Y$ be finite sets, $\pi:X\times Y\to [0,1]$ a probability distribution and $V:A\times B\times X \times Y\to \{0,1\}$ a function.

\begin{defn}
 A \emph{(two-player) nonlocal game} is a tuple~$G=(X,Y,A,B,\pi,V)$ describing a scenario consisting of non-communicating players, Alice and Bob, interacting with a referee.
In the game, the referee samples a pair of questions $(x,y)\in X\times Y$ according to~$\pi$, sending question~$x$ to Alice and~$y$ to Bob.
Then, Alice (respectively\ Bob) returns answer~$a$ (respectively\ $b$) to the referee. The players win if $V(a,b|x,y)=1$, otherwise they lose.
\end{defn}

The players are not allowed to communicate during the game, but they can agree on a strategy beforehand. For each strategy, we obtain probabilities $p(a,b|x,y)$ of Alice and Bob answering $a,b$ given $x,y$. The collection of numbers $\{p(a,b|x,y)\}$ is called a \emph{correlation}. The winning probability for the game, using this strategy, is given by
\begin{align*}
    \omega(G,S)=\omega(G,p)=\sum_{a,b,x,y} \pi(x,y)V(a,b|x,y)p(a,b|x,y)
\end{align*}

\subsection{Strategies for nonlocal games}
There are different types of strategies, which we now describe. We start with classical strategies.

\begin{defn}
A \emph{classical strategy} for a nonlocal game~$G$ consists of
\begin{enumerate}
    \item a probability distribution~$\gamma\colon\Omega\to [0,1]$ on a (without loss of generality) finite probability space~$\Omega$, along with
    \item probability distributions $\{p_\omega(a|x) : x\in X,\omega\in \Omega\}$ with outcomes in~$A$ and $\{q_\omega(b|y):y\in Y,\omega\in \Omega\}$ with outcomes in~$B$.
\end{enumerate}
In this case, we have
\begin{align*}
    p(a,b|x,y)=\sum_{\omega\in \Omega}\gamma(\omega) \, p_\omega(a|x) \, q_\omega(b|y).
\end{align*}
\end{defn}
A special case of classical strategies are \emph{deterministic strategies}, in which the answers of the players are determined by the questions. In this case, the strategies are described by functions $f:X\to A$ and $g:Y\to B$, and Alice and Bob answer with $f(x)$ and $g(y)$ given $x$ and $y$, respectively.

In a quantum strategy, the players are allowed to share a quantum state and  perform local measurements. A \emph{(quantum) state} $\ket{\psi}$ is a unit vector in a Hilbert space $\mathcal{H}$. The measurements are described by \emph{positive operator valued measures} (POVMs), which consist of a family of positive operators $\{M_i \in B(\mathcal{H})\,|\, i \in [m]\}$ such that $\sum_{i=1}^m M_i= 1_{B(\mathcal{H})}$. If all positive operators are projections ($M_i=M_i^*=M_i^2$), then we call $\{M_i \in B(\mathcal{H})\,|\, i \in [m]\}$ with $\sum_{i=1}^m M_i= 1_{B(\mathcal{H})}$ a \emph{projective measurement} (PVM).

\begin{defn}
A \emph{quantum strategy} for a nonlocal game $G$ consists of
\begin{enumerate}
    \item a (without loss of generality) pure quantum state $\ket{\psi}\in \mathcal{H}_A\otimes \mathcal{H}_B$, along~with
    \item POVMs $\{\{P_{a}^x:a\in A\}: x\in X\}$ acting on $\mathcal{H}_A$ and POVMs $\{\{Q_{b}^y:b\in B\}: y\in Y\}$ acting on $\mathcal{H}_B$.
\end{enumerate}
Then
\begin{equation*}
    p(a,b|x,y)=\bra\psi P_{a}^x \otimes Q_{b}^y \ket\psi.
\end{equation*}
\end{defn}

Now, we will define non-signalling strategies. These strategies do not have to be physical, but the resulting probability distribution ensures that Alice does not use Bob's input for her answer and vice versa.

\begin{defn}\label{def:nonsignalling}
A \emph{non-signalling strategy} consists of a correlation $\{p(a,b|x,y)\}$ such that
\begin{align*}
    \sum_{a\in A} p(a,b|x,y)&=\sum_{a\in A} p(a,b|\tilde{x},y):=p(b|y), \\
    \sum_{b\in B} p(a,b|x,y)&=\sum_{b\in B} p(a,b|x,\tilde{y}):=p(a|b)
\end{align*}
for all $x,\tilde{x}\in X$ and $y, \tilde{y} \in Y$.
\end{defn}

It is natural to ask how well players can perform in a nonlocal game for a specific class of strategies. Therefore, one defines the value of a game as the highest winning probability for a given class of strategies.

\begin{defn}
For a nonlocal game $G$ and we define
\begin{align*}
    \omega_c(G)&=\sup_{S \text{ classical strategy}}\omega(G,S), \\\omega_q(G)&=\sup_{S \text{ quantum  strategy}}\omega(G,S), \\\omega_{ns}(G)&=\sup_{S \text{ non-signalling strategy}}\omega(G,S).
\end{align*}
We call $\omega_c(G)$ the \emph{classical value}, $\omega_q(G)$ the \emph{quantum value} and $\omega_{ns}(G)$ the \emph{non-signalling value} of the game.
\end{defn}
From the definition, it is immediate that
\begin{align*}
    \omega_c(G)\leq \omega_q(G) \leq \omega_{ns}(G).
\end{align*}
Note that there exist games for which all these values are different, for example the CHSH game \cite{chsh}.

The proof of our main theorem follows an operator-algebraic framework for self-testing developed in\cite{zhao24}. For finite sets $X$ and $A$, we use $\msA_{PVM}^{X,A}$ to denote the universal $C^*$-algebra generated by $\{p^x_a:a\in A,x\in X\}$, subject to the relations 
\begin{itemize}
    \item $(p^x_a)^2=(p^x_a)^*=p^x_a \quad \forall a\in A, x\in X$
    \item $\sum_{a\in A}p^x_a=1 \quad \forall x\in X$
\end{itemize}
Given a nonlocal game $G=(X,Y,A,B,\mu,V)$, for any strategy 
\begin{equation} \label{int2}
    S=(\ket{\psi}\in\cH_A\otimes\cH_B,\{P^x_a\},\{Q^y_b\}),
\end{equation}
by the universal property, there are unique $*$-representations 
$$
\pi_A:\msA_{PVM}^{X,A}\rightarrow\cB(\cH_A) \text{ and } \pi_B:\msA_{PVM}^{Y,B}\rightarrow\cB(\cH_B)
$$
such that 
\begin{equation}  \label{int3}
    \pi_A(p^x_a)=P^x_a \text{ and } \pi_B(q^y_b)=Q^y_b \quad \forall a\in A , b\in B, x \in X, y\in Y.
\end{equation}
In the following we reserve the symbols $p^x_a$ and $q^y_b$ for the generators for $\msA_{PVM}^{X,A}$ and $\msA_{PVM}^{Y,B}$ respectively, whereas $P^x_a$ and $Q^y_b$ will denote representations given by operators on Hilbert spaces as in (\ref{int3}). 
We call $\pi_A,\pi_B$ the representations associated with $S$ as in ($\ref{int2}$), and sometimes write
 \begin{equation} \label{int4}
     S=(\ket{\psi}\in \cH_A\otimes\cH_B,\pi_A,\pi_B )
 \end{equation}
for a strategy.

Given a nonlocal game $G=(X,Y,A,B,\mu,V)$, we define the \textit{game polynomial} as 
$$
\Phi_G:=\sum_{a,b,x,y}\mu(x,y)V(a,b|x,y)p^x_a\otimes q^y_b
$$ 
which is an element of $\msA_{PVM}^{X,A}\otimes \msA_{PVM}^{Y,B}$ that encodes all the information of $G$. The value at the strategy (\ref{int4}) is then given by 
$$
\omega(G,S)=\bra{\psi} \pi_A \otimes \pi_B (\Phi_G) \ket{\psi}.
$$

\section{Feige's game}

In this section, we will discuss the classical, quantum and non-signalling value of a game that appeared in \cite{Feige}. The game was considered as counterexample in parallel repetition, as its classical value does not decrease when played twice in parallel. As main result of this section, we will show that its quantum value exceeds the classical value.

We recall again the definition of the game from the introduction. Alice and Bob both receive a bit as question (with uniform probability), and they are each allowed to choose an answer from the set $\{0,1,\perp\}$. To win the game, exactly one of the players has to answer with $\perp$, while the second player has to guess the input of the first player, i.e. the answer is a bit that is equal to the question of first player. More formally, the game $G_F$ \cite{Feige} is defined via
$$
    X= Y=\{0,1 \}, \quad A=B=\{0,1,\perp \}, \quad \pi(x,y)=\tfrac{1}{4} \quad  \forall (x,y)
$$
and has the predicate
$$
    V(a,b,x,y)=\begin{cases}
        1 \quad (a,b)=(\perp,x) \text{ or } (a,b)=(y,\perp),
        \\
        0 \quad \text{otherwise.}
    \end{cases}
$$
We will now argue that the classical value of this game is $\frac{1}{2}$. It is well-known that the classical value of a game is attained by a deterministic strategy. It is easy to see that the strategy in which Alice always answers $\perp$ and Bob always answers $0$ has winning probability $\frac{1}{2}$, since Alice receives the question $0$ halve of the time.

Note that in a deterministic strategy, winning or losing just depends on the input of the players. We have the following observations.

\begin{itemize}
    \item[(i)] Once a player answers a bit given a question (e.g. Alice answers $0$ when receiving $0$), the game will always be lost if the other player receives a question different from this bit (the referee sends $1$ to Bob), which happens halve of the time.
    \item[(ii)] In a deterministic strategy, the case that exactly one player answers $\perp$ and the other player answers with a bit is only possible for zero, two or all four different input pairs $\{(0,0), (0,1), (1,0), (1,1)\}$. If this happens for zero or two input pairs, the winning probability is upper bounded by $\frac{1}{2}$. If it is true for all input pairs, we get by (i) that the winning probability of the strategy is $\frac{1}{2}$.
\end{itemize}
Summarizing, we obtain $\omega_c(G_F)=\frac{1}{2}$. In the following, we will show that there exists a quantum strategy that allows the players to win the game with a higher probability. We will see that the winning probability of the quantum strategy is $\frac{9}{16}$. We show that this is an optimal quantum strategy for Feige's game in \Cref{appedix: SOS decomp.} and obtain a self-testing statement in \Cref{sect:SelftestFeige}.

\subsection{Construction of a quantum strategy exceeding the classical value}\label{sec:quantum_strategy}

The game polynomial for $G_F$ is given by

$$
\beta(p^x_a,q^y_b) = \tfrac{1}{4} \sum_{x,y\in \bit} (p^x_{\perp} \otimes q^y_x + p^x_y \otimes q^y_{\perp}).
$$

Consider the operators
\begin{align*}
    \tilde Z = Z \op 1 \quad\text{and}\quad \tilde X = X \op 1 \quad\text{on}\quad \C^3 = \C^2 \op \C
\end{align*}
(i.e., $\C^2$ is the defining and $\C$ the trivial representation of the Pauli group).
Note that $\tilde Z$ and $\tilde X$ are two non-commuting observables such that $\tilde Z$ and $\tilde X \tilde Z \tilde X$ commute (and every such pair of projections essentially looks like this, up to multiplicities).
Choose $\ket\bot \in \Fix(\tilde Z) {\color{gray}\ = \ket 1^\perp}$, for now arbitrarily, and extend it by $\ket{\tilde0}, \ket{\tilde 1}$ to a basis of~$\C^3$.
Alice's and Bob's strategy is then given by, 
\begin{equation*}
    \tilde{S}_F=(\ket{\psi}\in \C^3\otimes\C^3,\{P^x_a\},\{Q^y_b\})
\end{equation*}
where $\ket\psi=\frac{1}{\sqrt{3}}\left(\ket{00}+\ket{11}+\ket{22}  \right)$ and
\begin{align*}
    P^x_a &= \tilde X {\tilde Z}^x \tilde X \ket{\tilde a}\bra{\tilde a} \tilde X {\tilde Z}^x \tilde X &&(a \in \bit) \\
    P^x_\bot &= \tilde X {\tilde Z}^x \tilde X \ket\bot\bra\bot \tilde X {\tilde Z}^x \tilde X \\
    Q^y_b &= \tilde X P^y_b \tilde X &&(b \in \{0,1,\bot\})
\end{align*}
The winning probability is given by
\begin{align*}
    \omega(G_F,\tilde{S}_F)
&= \frac14 \sum_{x,y\in\bit} \frac13 \Tr\bracks{ P^x_y Q^y_\bot + P^x_\bot Q^y_x }
= \frac16 \sum_{x,y\in\bit} \Tr\bracks{ P^x_y Q^y_\bot } \\
&= \frac16 \sum_{x,y\in\bit} \Tr\bracks{ \tilde X {\tilde Z}^x \tilde X \ket{\tilde y}\bra{\tilde y} \tilde X {\tilde Z}^x \tilde X {\tilde Z}^y \tilde X \ket\bot\bra\bot \tilde X {\tilde Z}^y } \\
&= \frac13 \sum_{y\in\bit} \bra\bot \tilde X {\tilde Z}^y \ket{\tilde y}\bra{\tilde y} {\tilde Z}^y \tilde X \ket\bot
\end{align*}
where the first line holds because $\tilde X^2 = I$, and in the last step we used that $\tilde X \tilde Z \tilde X$ and $\tilde Z$ commute, as well as that $\tilde Z \ket\bot = \ket\bot$.
Now, we can write
\begin{align*}
    \ket\bot = \sqrt p \underbrace{\ket{\text{nontriv}}}_{\color{gray} = \ket 0} + \sqrt{1-p} \underbrace{\ket{\text{triv}}}_{\color{gray} = \ket 2},
\end{align*}
for some $p \in [0,1]$, where we have $\tilde X \ket{\text{triv}} = \tilde Z \ket{\text{triv}} = \ket{\text{triv}}$, while $\tilde Z \ket{\text{nontriv}} = - \tilde X \tilde Z \tilde X \ket{\text{nontriv}} = \ket{\text{nontriv}}$.
Then the above winning probability can be computed as follows
\begin{align*}
    \omega(G_F,\tilde{S}_F)
&= \frac p 3 \underbrace{\sum_{y\in\bit} \bra{\text{nontriv}} \tilde X {\tilde Z}^y \ket{\tilde y}\bra{\tilde y} {\tilde Z}^y \tilde X \ket{\text{nontriv}}}_{= 1} \\
&\quad + \frac{2 \sqrt{p(1-p)}}3 \sum_{y\in\bit} \bra{\text{triv}} \tilde X {\tilde Z}^y \ket{\tilde y}\bra{\tilde y} {\tilde Z}^y \tilde X \ket{\text{nontriv}} \\
&\quad + \frac{1-p}3 \underbrace{\sum_{y\in\bit} \bra{\text{triv}} \tilde X {\tilde Z}^y \ket{\tilde y}\bra{\tilde y} {\tilde Z}^y \tilde X \ket{\text{triv}}}_{= 1 - (1 - p) = p} \\
&= \frac{p(2-p)}3 + \frac{2 \sqrt{p(1-p)}}3 \bra{\text{triv}} \underbrace{\parens*{ \sum_{y\in\bit} (-1)^y \ket{\tilde y}\bra{\tilde y} } \tilde X \ket{\text{nontriv}}}_{= \ (\star)},
\end{align*}
where we used that 
$$
\tilde X \ket{\text{nontriv}} \in \Fix(-\tilde Z) = \Fix(\tilde Z)^\perp \subseteq \ket\bot^\perp.
$$
Note that 
$$
R = \sum_{y\in\bit} (-1)^y \ket{\tilde y}\bra{\tilde y}
$$
is an arbitrary reflection supported on $\ket\bot^\perp$.
Thus, $(\star)$ is an arbitrary state in $\ket\bot^\perp$.
It follows that the winning probability is maximized if we choose $R$ so that it maps $\tilde X \ket{\text{nontriv}}$ onto the orthogonal projection of $\ket{\text{triv}}$ onto $\ket\bot^\perp$.
Specifically, we can choose $R$ so that
\begin{equation}\nonumber
    \begin{split}
        \bra{\text{triv}}R\tilde{X} \ket{\text{nontriv}} = & \norm{(\1-\dyad{\bot}{\bot}) \ket{\text{triv}}}
        \\ = & \norm{p\ket{\text{triv}} + \sqrt{p(1-p)} \ket{\text{nontriv}}}
        \\ = & \sqrt{p}.
    \end{split}
\end{equation}
In this case, the winning probability is
\begin{align*}
  \omega(G_F,\tilde{S}_F)
= \frac p 3 \parens*{ 2-p + 2 \sqrt{1-p} }.
\end{align*}
This is maximized for $p=\frac34$, with corresponding success probability
\begin{align*}
  \omega(G_F,\tilde{S}_F) = \frac 9 {16}.
\end{align*}
Concretely, this means that $\ket{\text{nontriv}} = \ket 0$, $\ket{\text{triv}} = \ket 2$,
\begin{align*}
    \ket{\bot} = \sqrt{\frac34} \ket 0 + \sqrt{\frac14} \ket 2,
\end{align*}
we can pick the reflection $R$ to map $\tilde X \ket{\text{nontriv}} = \ket 1$ to the orthogonal projection of $\ket{\text{triv}} = \ket2$ onto $\ket\bot^\perp$, i.e., to $- \sqrt{\frac14} \ket 0 + \sqrt{\frac34} \ket 2$, while fixing $\ket\bot$.
Note that we can write
\begin{align*}
    R = \ket{\tilde 0}\bra{\tilde 0} - \ket{\tilde 1}\bra{\tilde 1},
\end{align*}
where
\begin{align*}
    \ket{\tilde 0} &= \frac1{\sqrt2} \ket1 - \frac1{2\sqrt2} \ket 0 + \frac{\sqrt3}{2\sqrt2} \ket 2 \\
    \ket{\tilde 1} &= \frac1{\sqrt2} \ket1 + \frac1{2\sqrt2} \ket 0 - \frac{\sqrt3}{2\sqrt2} \ket 2 \\
\end{align*}
and thus we have determined the strategy.

\begin{figure}[h]
    \centering
    \includegraphics[width=0.5\linewidth]{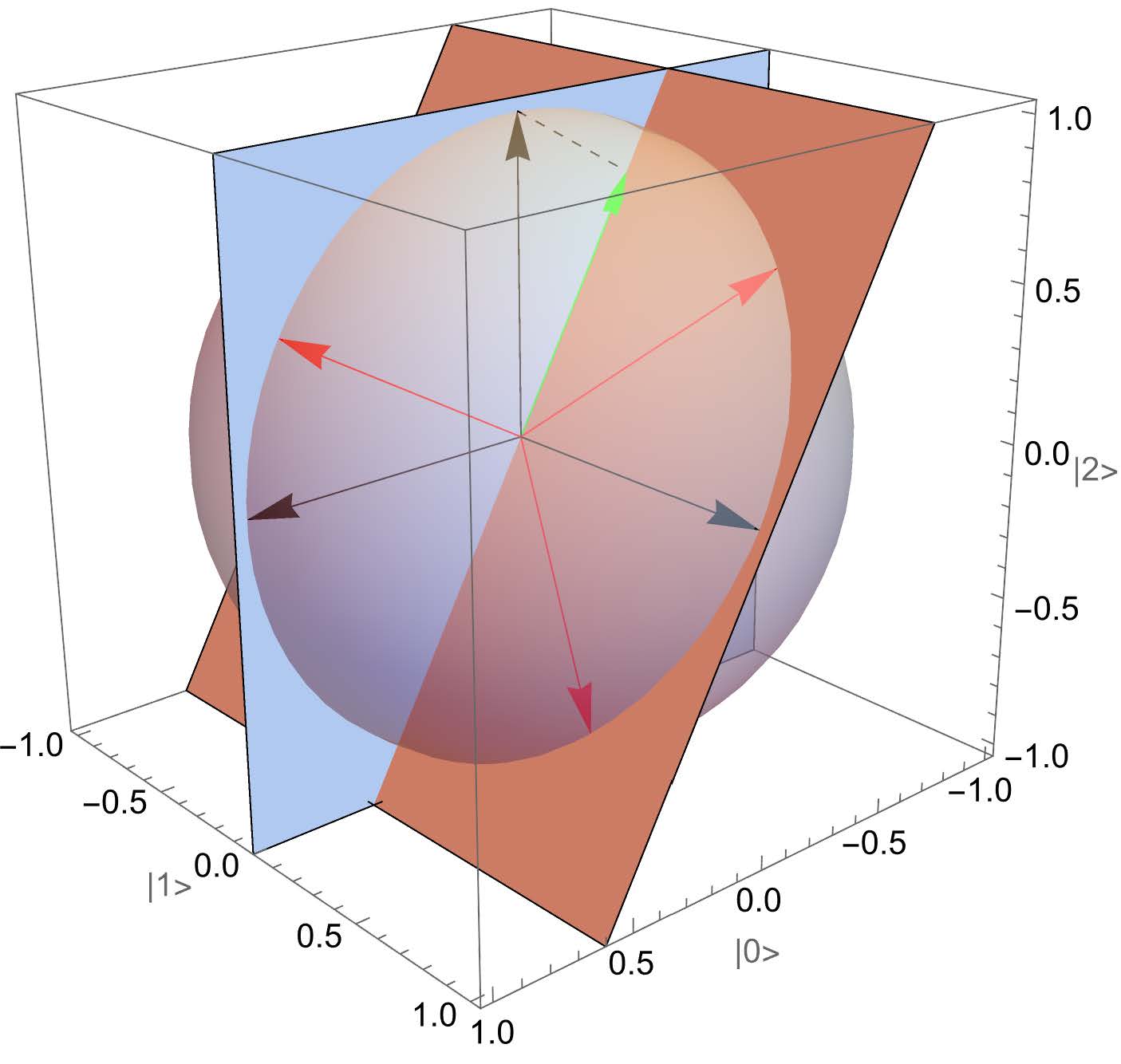}
    \caption{The red basis in this figure is $\{\ket{\tilde{0}}, \ket{\tilde{1}}, \ket \bot   \}$, the orange plane is $\ket \bot^\perp$ and the blue plane is $\ket 1^\perp$. The green vector is $(\1-\dyad{\bot}{\bot}) \ket{2}$.}
    \label{fig:enter-label}
\end{figure}

\subsection{Method for obtaining an upper bounds and algebraic relations} \label{sec:Upper bound}
In order to prove that the above strategy is optimal, we obtain a sum-of-squares (SOS) decomposition using a method similar to that of the NPA hierarchy (\cite{navascues2008convergent}). Let us start by defining the following sets of monomials 
$$
A^sB^t=\{ p^{x_1}_{a_1} ... p^{x_s}_{a_s} q^{y_1}_{b_1} ... q^{y_t}_{b_t} \mid x_1,...,x_s\in X, a_1,...,a_s\in A, y_1,...,y_t\in Y, b_1,...,b_t\in B  \}
$$
%
%
where $s,t\in \N_0$. 
We denote a level of the hierarchy by
$$
k+A^{s_1}B^{t_1}+...+A^{s_m}B^{t_m} : = \bigcup_{s+t\leq k} A^sB^t \cup \bigcup_{i=1}^m A^{s_i}B^{t_i}
$$
where $k\in \N$ and $m\in \N_0$ such that $s_i+t_i>k$ for all $i\in \{1,...,m\}$. 

The method can now be described as follows.

\begin{itemize}
    \item[(1)] Specify a level as described above, denoted $L$, and a vector of polynomials,
    $$
    F(p^x_a, q^y_b) = \left( f_1(p^x_a, q^y_b) , ... f_n(p^x_a, q^y_b) \right)^T
    $$
    where $n\in \N$ and for each $i\in \{1, ... , n \}$, $f_i(p^x_a, q^y_b)$ contains monomials from the set $L$. 
    \item[(2)] Specify an $n$-by-$n$ positive definite real matrix, $Y$, such that
    \begin{equation}\label{int1}
    F^* Y F = \mu_L - \beta.
    \end{equation}
    where $\mu_L \geq 0$.
\end{itemize}
Notice that $n$ may be strictly smaller than the cardinality of $L$. Notice also that we require $Y$ to be positive definite and not just positive semi-definite. 
Now, let 
$$
S=(\ket{\psi}\in \cH_A\otimes \cH_B,\pi_A, \pi_B)
$$ 
be a quantum strategy. Since $Y$ is positive definite it has a spectral decomposition
$$
Y = \sum_{i=1}^n \lambda_i \dyad{\varphi_i}{\varphi_i}
$$
where $\lambda_i > 0$ and the $\ket{\varphi_i}$'s constitute an orthonormal basis of $\C^n$. Let $\tilde{F}_i := \sqrt{\lambda_i} \bra{\varphi_i} F$ which are polynomials. Then (\ref{int1}) can be written as
$$
\sum_{i=1}^n \tilde{F}_i^*\tilde{F}_i = \mu_L - \beta.
$$
so we get
$$
\sum_{i=1}^n \pi_A \otimes \pi_B \left( \tilde{F}_i^*\tilde{F}_i \right) = \mu_L I- \pi_A \otimes \pi_B(\beta) \succeq 0.
$$
Thus, we obtain the upper bound
$$
\omega(G_F,S) = \bra{\psi} \pi_A \otimes \pi_B (\beta) \ket{\psi} \leq \mu_L.
$$
We can also deduce that certain algebraic relations must hold if $S$ in fact reaches the upper bound $\mu_L$. Suppose $\omega(G_F,S)=\mu_L$. Then
$$
\bra{\psi}\pi_A \otimes \pi_B \left( \tilde{F}_i^*\tilde{F}_i \right) \ket{\psi}  = 0 \quad \forall i\in \{1,...,n \}
$$
which implies
$$
\bra{\psi}\pi_A \otimes \pi_B \left( F^* \dyad{\varphi_i}{\varphi_i} F \right) \ket{\psi}=0 \quad \forall i\in \{1,...,n \}
$$
since $\lambda_i>0$. Now, summing over $i$ we obtain
$$
\bra{\psi}\pi_A \otimes \pi_B \left( F^* F \right) \ket{\psi} = \sum_{i=1}^n  \bra{\psi} \pi_A \otimes \pi_B \left( f_i^* f_i \right) \ket{\psi} = \sum_{i=1}^n \norm{ \pi_A \otimes \pi_B \left( f_i \right) \ket{\psi} }^2= 0.
$$
Thus, we get
$$
\pi_A \otimes \pi_B \left( f_i \right) \ket{\psi}=0 \quad \forall i\in \{1,...,n \}.
$$

Our approach is as follows. First, we use standard NPA techniques to verify the upper bound. In the notation here, this means letting the coordinate functions $f_i$ in $F$ be the elements of $L$ for some specific level of hierarchy and solving the following semi-definite program
\begin{equation}\label{standardNPA}
    \begin{split}
         \min & \text{ }  \mu_L 
        \\
        \text{such that } & F^* \tilde{Y} F = \mu_L - \beta
        \\ & \tilde{Y}  \succeq 0
    \end{split}
\end{equation}
From the output of this we are able to guess the relations that hold at optimality. Concretely, this is done by performing Gaussian elimination on the matrix consisting of the eigenvectors of $\tilde{Y}$ associated with nonzero eigenvalues. We then use these relations as in our polynomial basis $F$. Then we solve the following feasibility problem
\begin{equation}\label{eq:feas_prob}
    \begin{split}
         \min & \text{ }  0
        \\
        \text{such that } & F^* \tilde{Y} F = \mu_L - \beta
        \\ & \tilde{Y} - \varepsilon I  \succeq 0
    \end{split}
\end{equation}
for a small $\varepsilon>0$ where $\mu_L$ is fixed as the output of (\ref{standardNPA}). The output of (\ref{eq:feas_prob}) we then use as our matrix $Y$ in (2) above. 

In the case of Feige's game, in the standard NPA approach, using level $1+AB$ suffices to determine that $\frac{9}{16}$ is an upper bound on $\omega(G_F)$. However, in this work, in order to achieve the desired algebraic relations, we take $L$ to be $1+AB+A^2+A^3$ and $F$ as specified in \Cref{appedix: SOS decomp.}. Moreover we are also able to round the numerical output to an exact solution, $Y$, with rational entries satisfying
$$
F^* Y F = \frac{9}{16} - \beta.
$$
We have stored this exact solution in a Mathematica notebook in a GitHub repository \cite{repository}. Running the notebook verifies the claims made in this paper.

\subsection{The non-signalling value of Feige's game}

We have seen that the classical value of Feige's game is $\frac{1}{2}$, and we can exceed this value with a quantum strategy that wins with probability $\frac{9}{16}$. In this section, we will study the non-signalling value of $G_F$ and prove $\omega_{ns}(G_F)=\frac{2}{3}$. Recall the definition of non-signalling strategies from \Cref{def:nonsignalling}.

We start with the description of a non-signalling strategy that achieves the value $\frac{2}{3}$. Consider the non-signalling strategy $p$, where
\begin{align*}
    p(a,b|x,y)=\begin{cases}
        \frac{1}{3} \quad\text{ if }(a,b)=(\perp,x) \text{ or } (a,b)=(y,\perp),\\
        \frac{1}{3} \quad\text{ if } a\neq y, b\neq x \text{ and } a,b \in \{0,1\},\\
        0 \quad \text{ otherwise.}
    \end{cases}
\end{align*}
We see that $p$ is a probability distribution and it holds
\begin{align*}
    \sum_a p(a,\perp|x,y)&= p(y,\perp|x,y)=\frac{1}{3},\\
    \sum_a p(a,x|x,y)&= p(\perp,x|x,y)=\frac{1}{3},\\
    \sum_a p(a,\overline{x}|x,y)&=p(\overline{y},\overline{x}|x,y)=\frac{1}{3},
\end{align*}
where $\overline{x}$ is the complement of the bit $x$. This shows $\sum_a p(a,b|x,y)=\sum_a p(a,b|x',y)$ and similarly $\sum_b p(a,b|x,y)=\sum_b p(a,b|x,y')$ holds. We conclude that $p$ is a non-signalling strategy. We obtain
\begin{align*}
    \omega(G_F,p)&=\frac{1}{4}\sum_{x,y\in\{0,1\}}p(\perp,x|x,y)+p(y,\perp|x,y)=\frac{2}{3},
\end{align*}
since $p(\perp,x|x,y)=p(y,\perp|x,y)=\frac{1}{3}$.

For the upper bound of the winning probability of non-signalling strategies, consider any non-signalling strategy $p$ for $G_F$. We first note that
\begin{align}\label{Eq:nonsig1}
   p(\perp,x|x,y)+p(\perp,\overline{x}|\overline{x},y)+p(y,\perp|x,y) &\leq \sum_a p(a,x|x,y)+p(a,\overline{x}|\overline{x},y)+p(a,\perp|x,y)\nonumber\\
   &=\sum_a p(a,x|x,y)+p(a,\overline{x}|x,y)+p(a,\perp|x,y)\nonumber\\
   &= 1,
\end{align}
by the non-signalling property of $p$ and since it is a probability distribution. Similary, it holds
\begin{align}\label{Eq:nonsig2}
    p(\perp,x|x,y)+p(\overline{y},\perp|x,\overline{y})+p(y,\perp|x,y)\leq 1.
\end{align}
We compute
\begin{align*}
    12 \omega(G_F,p)&=3\sum_{x,y\in\{0,1\}}p(\perp,x|x,y)+p(y,\perp|x,y)\\
    &=2\sum_{x,y\in\{0,1\}}p(\perp,x|x,y)+p(y,\perp|x,y)+\sum_{x,y\in\{0,1\}}p(\perp,\overline{x}|\overline{x},y)+p(\overline{y},\perp|x,\overline{y})\\
    &=\sum_{x,y\in\{0,1\}}p(\perp,x|x,y)+p(\perp,\overline{x}|\overline{x},y)+p(y,\perp|x,y)\\
    &\hspace{1.5cm} +\sum_{x,y\in\{0,1\}}p(\perp,x|x,y)+p(\overline{y},\perp|x,\overline{y})+p(y,\perp|x,y)\\
    &\leq 8,
\end{align*}
by rearranging terms and using \Cref{Eq:nonsig1,Eq:nonsig2} for the inequality. This yields $\omega(G_F,p)\leq \frac{2}{3}$. Summarizing, we obtain $\omega_{ns}(G_F)=\frac{2}{3}$.





\section{Feige's game is a self-test}\label{sect:SelftestFeige}

The goal of this section is to show that $\tilde{S}_F$ given in \Cref{sec:quantum_strategy} is in fact the unique optimal quantum strategy of $G_F$ in the sense made precise in the following definition.

\begin{defn} \label{def:self-test}
    Let $G=(X,Y,A,B,\mu,V)$ be a nonlocal game, and let
    \begin{equation*}
        \tilde{S}=(\ket{\tilde{\psi}}\in \tilde{\cH}_A\otimes\tilde{\cH}_B, \{\tilde{P}^x_a\} , \{\tilde{Q}^y_b\})
    \end{equation*}
    be an optimal quantum strategy. We say that $G$ \textit{robustly self-tests $\tilde{S}$} if there exists a function 
    \begin{equation*}
        \delta:\R_{\geq 0}\rightarrow\R_{\geq 0}, \qquad \delta(\epsilon)\rightarrow 0 \text{ as } \epsilon\rightarrow 0,
    \end{equation*}
such that for any $\epsilon\geq 0$ and any strategy
    \begin{equation*}
       S= (\ket{\psi}\in \cH_A\otimes\cH_B, \{P^x_a\} , \{Q^y_b\}),
    \end{equation*}
    with $\omega(G,S)=\omega_q(G)-\epsilon$, 
    there exist isometries 
    $$
    V_A:\cH_A\rightarrow \tilde{\cH}_A\otimes \cK_A \text{, } V_B:\cH_B\rightarrow \tilde{\cH}_B\otimes \cK_B
    $$ 
    and a unit vector $\ket{\kappa}\in \cK_A\otimes \cK_B$ such that
    \begin{equation}
       \norm{ V_A\otimes V_B \left(P^x_a\otimes Q^y_b\ket{\psi}\right)-\left(   \tilde{P}^x_a\otimes   \tilde{Q}^y_b \ket{\tilde{\psi}} \right) \otimes \ket{\kappa} }\leq \delta(\epsilon) \label{eq:close}
    \end{equation}
    for all $x,y,a,b$.  
\end{defn}

%
%
Recall the following definition from \cite[Definitions 6.1 and 7.1]{zhao24}.
\begin{defn}\label{def:pair}
    Let $\Gamma=\{\gamma^y_b:b\in B,y\in Y\}$ and $\cR$ be two sets of $*$-polynomials in $\msA_{PVM}^{X,A}$, where every $\gamma^y_b$ is self-adjoint. We say $(\Gamma,\cR)$ is a \textit{determining pair} for a nonlocal game $G=(X,Y,A,B,\mu,V)$ if the following conditions hold.
    \begin{enumerate}
        \item For every $y\in Y$, $\{q_{\cR}(\gamma^y_b):b\in B\}$ is a PVM in the quotient $\msA_{PVM}^{X,A}/\langle \cR\rangle$, where $\langle \cR\rangle$ is the  two-sided closed ideal generated by $\cR$ and 
        $$
        q_{\cR}:\msA_{PVM}^{X,A}\rightarrow \msA_{PVM}^{X,A}/\langle \cR\rangle
        $$ 
        is the quotient map.
        \item A strategy $\tilde{S}=(\ket{\tilde{\psi}},\tilde{\pi}_A,\tilde{\pi}_B )$ is optimal if and only if
        \begin{enumerate}
            \item $I_A\otimes \tilde{\pi}_B(q^y_b)\ket{\tilde{\psi}}=\tilde{\pi}_A(\gamma^y_b)\otimes I_B\ket{\psi}$ for all $b,y$, and
            \item the linear functional $\tau$ on $\msA_{PVM}^{X,A}$ defined by 
            $$
            \tau(\alpha):=\bra{\tilde{\psi}}\tilde{\pi}_A(\alpha)\otimes I_B\ket{\tilde{\psi}}
            $$  
            is a tracial state satisfying $\tau(r^*r)=0$ for all $r\in\cR$.
        \end{enumerate}
    \end{enumerate}
    In this case, we define the \textit{game algebra of $G$} as 
    $$
    C^*(G):=\msA_{PVM}^{X,A}/\langle \cR\rangle.
    $$
    We say that a determining pair $(\Gamma,\cR)$ is \textit{$\delta$-robust} for some function  
    \begin{equation*}
        \delta:\R_{\geq 0}\rightarrow\R_{\geq 0}, \qquad \delta(\epsilon)\rightarrow 0 \text{ as } \epsilon\rightarrow 0
    \end{equation*}
    if for any $\epsilon\geq 0$ and any quantum strategy $S=(\ket{\psi},\pi_A,\pi_B)$ with $\omega(G,S)=w_q(G)-\epsilon$,
    \begin{enumerate}
        \item[(3)] $\norm{I_A\otimes \pi_B(q^y_b)\ket{\psi} - \pi_A(\gamma^y_b)\otimes I_B\ket{\psi}  }\leq \delta(\epsilon)$ for all $b,y$, and
        \item[(4)] $\norm{ \pi_A(r)\otimes I_B\ket{\psi}}\leq \delta(\epsilon)$ for all $r\in \cR$.
    \end{enumerate}
\end{defn}

The optimal strategy $\tilde{S}$ in \Cref{def:pair} is sometimes said to be \textit{Alice-determined}: Bob's PVMs $\{q^y_b\}$ are determined by some PMVs $\Gamma=\{\gamma^y_b\}$ on Alice's side (see (2a)), and Alice's PVMs satisfy the algebraic relations $\cR$ (see (2b)). A nearly optimal strategy is then \textit{Alice-robustly-determined}, meaning that the equalities in (2a) and (2b) still hold approximately. In particular, Alice's PVMs in an optimal (resp. nearly optimal) strategy define a representation (resp. an approximate representation) of the associated game algebra $C^*(G)$. 

Thus, to prove $G$ has a unique optimal strategy in the sense of definition \Cref{def:self-test}, it suffices to show that $C^*(G)$ has a unique irreducible representation, and the robustness of this self-test then follows from the robustness of the determining pair.

\begin{thm}[Corollary 7.11 of \cite{zhao24}]\label{uniqueirrep-selftest}
    Let $G$ be a nonlocal game with a $\delta$-robust determining pair $(\Gamma,\cR)$. If the associated game algebra $C^*(G)$ is a full matrix algebra, then $G$ is a robust self-test.
\end{thm}

Within this framework, to prove that $G_F$ robustly self-tests $\tilde{S}_F$, we need to show that
\begin{enumerate}
    \item[(A)]  $G_F$ has a robust determining pair $(\Gamma,\cR)$, and
    \item[(B)]  the associated game algebra $C^*(G_F)$ has a unique irreducible representation.
\end{enumerate}

Our approach for Task (A), constructing a determining pair and showing its robustness, is based on a sum-of-squares (SOS) decomposition of the game polynomial of $G_F$ as detailed in \Cref{sec:Upper bound}. 
%
%
We get the following algebraic structure of optimal strategies for $G_F$.
\begin{lem}\label{lem:pair}
  Let $ \nu := \begin{bmatrix}
 \frac{9}{16} & \frac{1}{16} & \frac{3}{8} \\
 \frac{1}{16} & \frac{9}{16} & \frac{3}{8}  \\
 \frac{3}{8} & \frac{3}{8} & \frac{1}{4}
\end{bmatrix}$, with the first, second, and third column/row labelled by $0,1$, and $\bot$, respectively. Let 
$$
\Gamma:=\{\gamma^y_b: b\in\{ 0,1,\bot\},y\in\{0,1\} \}\subseteq \msA_{PVM}^{\{0,1\},\{ 0,1,\bot\}}
$$ 
where
\begin{align*}
    \gamma^0_0 : &=-\tfrac{3}{4} \{p^0_1 , p^1_0 \} - \tfrac{5}{4}\{p^0_1 , p^1_1 \}
-\tfrac{5}{16} p^0_0+\tfrac{3}{8} p^1_0 +\tfrac{9}{16} p^0_1+\tfrac{15}{8}p^1_1,\\
\gamma^0_1 :&=-\tfrac{3}{4} \{p^0_1 , p^1_0 \} - \tfrac{5}{4}\{p^0_1 , p^1_1 \}
+\tfrac{3}{16} p^0_0-\tfrac{1}{8} p^1_0 +\tfrac{33}{16} p^0_1+\tfrac{3}{8}p^1_1,\\
\gamma^0_\bot :&=1 - \gamma^0_0- \gamma^0_1,\\
\gamma^1_0 :&=-\tfrac{3}{4} \{p^0_1 , p^1_0 \} - \tfrac{5}{4}\{p^0_1 , p^1_1 \}
-\tfrac{9}{16} p^0_0+\tfrac{9}{8} p^1_0 +\tfrac{13}{16} p^0_1+\tfrac{9}{8}p^1_1,\\
\gamma^1_1 :&=-\tfrac{3}{4} \{p^0_1 , p^1_0 \} - \tfrac{5}{4}\{p^0_1 , p^1_1 \}
+\tfrac{15}{16} p^0_0-\tfrac{3}{8} p^1_0 +\tfrac{21}{16} p^0_1+\tfrac{5}{8}p^1_1,\\
\gamma^1_\bot :&=1 - \gamma^1_0- \gamma^1_1.
\end{align*}
Here $\{p,q\}:=pq+qp$ is the anti-commutator. Let
\begin{equation*}
    \cR:=\{p^x_a p^{1-x}_{a'} p^x_a- \nu_{aa'} p^x_a: a,a'\in \{ 0,1,\bot\},  x\in \{0,1 \} \}\subseteq \msA_{PVM}^{\{0,1\},\{ 0,1,\bot\}}.
\end{equation*}
Then $(\Gamma,\cR)$ is an $O(\sqrt{\epsilon})$-robust determining pair of $G_F$. 
\end{lem}
\begin{proof}
It is not hard to verify that $\gamma^y_b$'s are self-adjoint such that $\sum_{b}\gamma^y_b=1$ for $y\in \{0,1\}$, and $(\gamma^y_b)^2-\gamma^y_b\in \langle\cR\rangle$ for all $y,b$. So for every $y$, $\{\gamma^y_b:b\in\{0,1,\bot\}\}$ is a PVM in $\msA_{PVM}^{\{0,1\},\{ 0,1,\bot\}}/\langle\cR\rangle.$
This proves Hypothesis (1) in \Cref{def:pair}.

Let $\beta \in \msA_{PVM}^{X,A}\otimes \msA_{PVM}^{Y,B}$ be the game polynomial of $G_F$, and let 
$$
F=F(p_a^x, q_b^y)=\left(f_1(p_a^x, q_b^y), ... ,f_{32}(p_a^x, q_b^y)\right)^T
$$ 
be the vector of 32 polynomials in $p_{a}^x$ and $q_b^y$ given in \Cref{appedix: SOS decomp.}. Let $Y$ be the $32\times 32$ positive definite matrix given in \cite{repository}, then as stated in \Cref{sec:Upper bound}, using a computer, one can check that
\begin{equation*}
    F^* Y F = \tfrac{9}{16}-\beta.
\end{equation*}
%
As in \Cref{sec:Upper bound} we write $Y\in M_{32}(\C)$ in its spectral decomposition
\begin{equation*}
    Y= \sum_{i=1}^{32} \lambda_i \dyad{\varphi_i}{\varphi_i},
\end{equation*}
where $\lambda_i>0$ for all $i$ and $\{\ket{\varphi_i}:1\leq i\leq 32\}$ is an orthonormal basis of $\C^{32}$. For any quantum strategy $S=(\ket{\psi},\pi_A,\pi_B)$, let 
$$
\hat{F}=\big((\pi_A\otimes\pi_B)(f_1),\ldots, (\pi_A\otimes\pi_B)(f_{32})\big)^T,
$$ 
and let $\epsilon=\frac{9}{16}-\omega(G_F,S)$. It follows that
\begin{align*}
    \epsilon&=\bra{\psi} (\pi_A\otimes\pi_B)\left(\frac{9}{16}-\beta \right)\ket{\psi}=\bra{\psi} \hat{F}^*Y\hat{F}\ket{\psi}\\
    &=\sum_{i=1}^{32}\lambda_i \bra{\psi} \hat{F}^* \ket{\varphi_i} \bra{\varphi_i}\hat{F} \ket{\psi} = \sum_{i=1}^{32} \lambda_i \abs{\bra{\varphi_i}\hat{F}\ket{\psi}}^2.
\end{align*}
Since each $\lambda_i>0$, we obtain that
\begin{equation*}
    \abs{\bra{\varphi_i}\hat{F}\ket{\psi}}^2 \leq O(\epsilon)
\end{equation*}
for all $1\leq i\leq 32$. Hence
\begin{equation*}
    \sum_{i=1}^{32} \norm{(\pi_A\otimes\pi_B(f_i)\ket{\psi} }^2= \bra{\psi}\hat{F}^*\hat{F}\ket{\psi}=\sum_{i=1}^{32} \abs{\bra{\varphi_i}\hat{F}\ket{\psi}}^2 \leq O(\epsilon).
\end{equation*}
In particular, 
$$
\norm{(\pi_A\otimes\pi_B)(f_i)\ket{\psi} }\leq O(\sqrt{\epsilon})
$$ 
for every $1\leq i\leq 32$. We further observe that 
$$
1\otimes q^y_b-\gamma^y_b\otimes 1,r\otimes 1\in \spa\{f_i:1\leq i\leq 32\}
$$ 
for all $b,y$ and all $r\in \cR$. We conclude that 
\begin{equation}
    \norm{I_A\otimes \pi_B(q^y_b)\ket{\psi} - \pi_A(\gamma^y_b)\otimes I_B\ket{\psi}  }\leq O(\sqrt{\epsilon}) \label{eq:approxGamma}
\end{equation}
for all $b,y$, and
\begin{equation}
    \norm{ \pi_A(r)\otimes  I_B\ket{\psi}}\leq O(\sqrt{\epsilon}) \label{eq:approxR}
\end{equation}
for all $r\in \cR$. This proves Hypotheses (3) and (4) in \Cref{def:pair} with $\delta(\epsilon)=O(\sqrt{\epsilon})$.

Now suppose $S$ is optimal for $G_F$, then $\epsilon=0$, so the ``only if" part of Hypothesis (2a) follows from \Cref{eq:approxGamma}. By \Cref{eq:approxR}, the linear functional $\tau$ on $\msA_{PVM}^{\{0,1\},\{ 0,1,\bot\}}$ defined by 
$$
\tau(\alpha):=\bra{\psi}\pi_A(\alpha)\otimes I_B\ket{\psi}
$$ 
is a state satisfying $\tau(r^*r)=0$ for all $r\in \cR$. By symmetry, for every $p^x_a$ there exists self-adjoint $\eta^x_a\in \msA_{PVM}^{Y,B}$ such that $\pi_A(p^x_a)\otimes I_B\ket\psi=I_A\otimes \pi_B(\eta^x_a)\ket{\psi}$. So for any monomials $w_1=p^{x_1}_{a_1}\cdots p^{x_k}_{a_k}$ and $w_2=p^{x_{k+1}}_{a_{k+1}}\cdots p^{x_n}_{a_n}$, we have
    \begin{align*}
        \bra\psi \pi_A(w_2w_1)\otimes I_B\ket\psi&= \bra\psi \pi_A(w_2)\pi_A(p^{x_1}_{a_1}\cdots p^{x_{k-1}}_{a_{k-1}})\otimes \pi_B(\eta^{x_k}_{a_k})\ket\psi\\
        &=\bra\psi \pi_A(p^{x_k}_{a_k})\pi_A(w_2)\pi_A(p^{x_1}_{a_1}\cdots p^{x_{k-1}}_{a_{k-1}})\otimes I_B\ket\psi\\
        &= \cdots\\
        &= \bra\psi\pi_A(p^{x_1}_{a_1}\cdots p^{x_k}_{a_k})\pi_A(w_2)\otimes I_B\ket\psi\\
        &= \bra\psi \pi_A(w_1w_2)\otimes I_B\ket\psi.
    \end{align*}
Hence $\tau$ is tracial. This proves the ``only if" part of (2b).

To see the ``if" part of Hypothesis (2), suppose $\tilde{S}=(\ket{\tilde{\psi}},\tilde{\pi}_A,\tilde{\pi}_B)$ is a strategy satisfying (2a) and (2b). Then $\bra{\tilde{\psi}}(\tilde{\pi}_A\otimes \tilde{\pi}_B)(\beta)\ket{\tilde{\psi}}=\frac{9}{16}$. Hence $\tilde{S}$ is optimal. This completes the proof.
\end{proof}

For Task (B), proving the game algebra $C^*(G_F)$ has a unique irreducible representation, we first note that by \Cref{lem:pair}, $C^*(G)=\msA_{PVM}^{\{0,1\},\{ 0,1,\bot\}}/\langle \cR\rangle$ is isometric to $\cA_\nu$, the universal $C^*$-algebra generated by $E_i,F_i,1\leq i\leq 3$, subject to the relations
\begin{enumerate}[label=(\roman*)]
    \item  $E_i = E_i^*$, $F_i^*=F_i$ for all $i$,
    \item $\sum_{i=1}^3 E_i =\sum_{i=1}^3 F_i = 1$, and
    \item $E_iF_jE_i=\nu_{ij} E_i$, $F_iE_jF_i=\nu_{ij}F_i$ for all $i,j$,
\end{enumerate}
where $\nu$ is the symmetric stochastic matrix defined in \Cref{lem:pair}. In fact, for any $n\times n$ symmetric stochastic matrix $\nu$, we can define $\cA_\nu$ analogously, where generators $\{E_i\}_{i=1}^n$ and $\{F_i\}_{i=1}^n$ are PVMs satisfying $E_iF_jE_i=\nu_{ij} E_i$ and $F_iE_jF_i=\nu_{ij}F_i$. In \Cref{sect:generalMUM}, we establish the representation theory for those $\cA_\nu$'s. Interestingly, we show that representations for $\cA_\nu$ can be seen as a generalization of mutually unbiased measurements (MUMs, see e.g. \cite{MUM_Tavakoli_etal}). We generalize the results in \cite{MUM_Tavakoli_etal} and give a characterization of irreducible representations of $\cA_\nu$. In particular, in \Cref{thm:irrep}, we show that for $2\times 2$ and $3\times 3$ symmetric stochastic matrices $\nu$, every $\cA_\nu$ has a unique irreducible representation up to unitary equivalence and relabeling generators. For the matrix $\nu$ defined in \Cref{lem:pair}, we further show that any allowed relabeling still results in a unitarily equivalent representation.
\begin{lem}\label{lem:unique}
    The game algebra $C^*(G_F)$ associated with $G_F$ has a unique irreducible representation up to unitary equivalence. Moreover, $G_F$ is isomorphic to the full matrix algebra $M_3$.
\end{lem}
We prove this lemma in \Cref{sect:generalMUM}.

\begin{thm}\label{thm:self-test}
   The game $G_F$ robustly self-tests the strategy $\tilde{S}_F$.
\end{thm}

\begin{proof}
  Let $(\Gamma,\cR)$ be the $O(\sqrt{\epsilon})$-robust determining pair of $G_F$ defined in \Cref{lem:pair}. By \Cref{lem:unique}, the associated game algebra  $C^*(G)=\msA_{PVM}^{\{0,1\},\{ 0,1,\bot\}}/\langle \cR\rangle$ is a full matrix algebra. It follows from \Cref{uniqueirrep-selftest} that $G_F$ is a robust self-test.


\end{proof}

\section[Feige's game as "or"-game]{Feige's game as $(G_1\lor G_2)$-game}

Let $G_1$ and $G_2$ be nonlocal games. We describe the $(G_1\lor G_2)$-game as introduced in \cite{mancinskaschmidt}. The referee sends Alice and Bob a pair of questions $(x_1, x_2)$ and $(y_1, y_2)$, respectively, where $x_i, y_i$ are questions in $G_i$, $i\in \{1,2\}$. Each of them chooses one of the questions they received and responds with an answer from the corresponding game. To win the game, two conditions have to be fulfilled:
\begin{itemize}
    \item[(1)] Alice and Bob have to give answers from the same game,
    \item[(2)] their answers have to win the corresponding game.
\end{itemize}
More formally, suppose the nonlocal games $G_1$ and $G_2$ have input sets $X_{i}$, $Y_{i}$, output sets $A_{i}$, $B_{i}$, verification functions $V_{G_i}$ and probability distributions $\pi_i$ on $X_{i}\times Y_{i}$ for $i=1,2$. Then the $(G_1\lor G_2)$-game has input sets $X_{1}\times X_{2}$, $Y_{1}\times Y_{2}$, output sets $A_{1} \dot\cup A_{2}$, $B_{1}\dot\cup B_{2}$ and verification function
\begin{align*}
    V((x_1, x_2), (y_1,y_2), a,b)=\begin{cases}V_{G_i}(x_i,y_i,a,b) \,\text{ if } a \in A_{i} \text{ and } b \in B_{i} \text{ for some $i=1,2$,} \\0 \, \text{ otherwise.}\end{cases}
\end{align*}
For the probability distribution, we take $\pi=\pi_1 \times \pi_2$ on $(X_{1}\times X_{2})\times (Y_{1}\times Y_{2})$, \emph{i.e.} $\pi((x_1,x_2),(y_1,y_2))=\pi_1(x_1,y_1)\pi_2(x_2,y_2)$. In this article, we assume that the probability distributions of $G_1$ and $G_2$ are uniform, so this will also be the case for the $(G_1\lor G_2)$-game.

We will now show that there exist $G_1$ and $G_2$ such that $G_F=G_1\lor G_2$, where $G_F$ denotes the Feige game. We will see $\frac{9}{16}=\omega_q(G_F)>\mathrm{max}\{\omega_q(G_1),\omega_q(G_2)\}=\frac{1}{2}$. This is in contrast the case of perfect strategies \cite[Lemma 4.3]{mancinskaschmidt}, where $\omega_q(G_1\lor G_2)=1$ implies $\omega_q(G_i)=1$ if the players give answers for $G_i$ a perfect strategy of the "or"-game.

Consider the games $G_{1}$ and $G_{2}$, where
\begin{align}
&X_1=\{*\}, \quad Y_1=\{0,1\},\quad  A_1=\{0,1\},\quad B_1=\{\perp\},\quad \pi(x,y)=\frac{1}{2}\label{eq1}\\
&X_2=\{0,1\},\quad Y_2=\{*\},\quad A_2=\{\perp\}, \quad B_2=\{0,1\},\quad \pi(x,y)=\frac{1}{2},\label{eq2}\\
&V_1(a,\perp,*,y)=\begin{cases}1 \text{ if } a=y\\0 \text{ otherwise,}\end{cases} V_2(\perp,b,x,*)=\begin{cases}1 \text{ if } b=x\\0 \text{ otherwise.}\end{cases}
\end{align}
In the game $G_1$, Alice guesses the bit that Bob received and their roles are reversed in the game $G_2$.

\begin{lem}
It holds $\omega_q(G_i)=\frac{1}{2}$ for $i=1,2$.
\end{lem}

\begin{proof}
We prove it for $G_1$, the proof for $G_2$ is the same. Note that it holds $B^y_{\perp}=\mathrm{Id}$ for each Bob POVM in any quantum strategy $S$, since he is only allowed to answer with $\perp$. Thus, the winning probability of any quantum strategy for $G_1$ is given by
\begin{align*}
\omega_q(S, G_1)=\frac{1}{2}\bra{\psi}A^*_{0}\otimes B^0_{\perp}+A^*_{1}\otimes B^1_{\perp}\ket{\psi}
=\frac{1}{2}\bra{\psi}(A^*_{0}+A^*_{1})\otimes \mathrm{Id}\ket{\psi}
=\frac{1}{2},
\end{align*}
where we used $A^*_{0}+A^*_{1}=\mathrm{Id}$ in the last step. This completes the proof.
\end{proof}

\begin{thm}
It holds $G_F=G_1\lor G_2$. We therefore have
\begin{align*}
\omega_q(G_F)=\frac{9}{16}>\frac{1}{2}=\mathrm{max}\{\omega_q(G_1),\omega_q(G_2)\}.
\end{align*}
\end{thm}

\begin{proof}
From \Cref{eq1,eq2}, we deduce
\begin{align*}
X_{G_1\lor G_2}&=\{(*,0),(*,1)\}, &&Y_{G_1\lor G_2}=\{(0,*),(1,*)\}, \\
A_{G_1\lor G_2}&=\{0,1,\perp\}, &&B_{G_1\lor G_2}=\{0,1,\perp\}.
\end{align*}
Now, relabeling the sets $\{(*,0),(*,1)\}$ and $\{(0,*),(1,*)\}$ by $\{0,1\}$, we obtain that $G_1\lor G_2$ has the same question and answer sets as the Feige game. It is straightforward to see that the verification functions coincide.
\end{proof}

\section{Feige's game and parallel repetition}

In this section, we look at parallel repetition of Feige's game. From \Cref{tab:tabintro}, we see that Feige's game has interesting behaviour with respect to parallel repetition. In the table, the classical, quantum and non-signalling values agree in the case of an even number of repetitions, but this is not the case for odd numbers. As we saw previously, all values are different for the original game. Furthermore, we get that at least the classical and non-signalling value are different for $n=3$, despite this not being the case for $n=2$. We do not know if there is quantum advantage for $n=3$. Note that the quantum strategy consisting of the optimal quantum strategy for one game and the optimal classical strategy for the $2$-fold parallel repetition performs worse than the optimal classical strategy for $n=3$, since $\frac{9}{16}\cdot\frac{1}{2}=\frac{9}{32}<\frac{5}{16}$.

We first recall the definition of parallel repetition. In the following, we define the product of two games, which then leads to the definition of the parallel repetition of a game.
\begin{defn}
Suppose
\begin{equation} \nonumber
    \begin{split}
        G_1 = (X_1,Y_1,A_1,B_1,\pi_1, V_1),
        \quad
        G_2 = (X_2,Y_2,A_2,B_2,\pi_2, V_2)
    \end{split}
\end{equation}
are nonlocal games. We define their \emph{product} as the game
\begin{equation}\nonumber
    G_1 \times G_2 = \big(X_1\times X_2,Y_1\times Y_2,A_1\times A_2,B_1\times B_2, \pi_1 \times \pi_2,V_1 \times V_2\big)
\end{equation}
Where Alice and Bob receive pairs of questions $(x_1,x_2)\in X_1\times X_2$ and $(y_1,y_2)\in Y_1\times Y_2$ respectively, and provide answers $(a_1,a_2)\in A_1 \times A_2$ and $(b_1,b_2 ) \in B_1 \times B_2$ respectively. The prior distribution is given by
\begin{equation}\nonumber
    \pi_1\times \pi_2 \left( (x_1,x_2),(y_1,y_2) \right):= \pi_1(x_1,y_1) \pi_2(x_2,y_2).
\end{equation}
and $V_1\times V_2$ likewise
\begin{equation}\nonumber
    V_1\times V_2 \Big( (a_1,a_2),(b_1,b_2)|(x_1,x_2),(y_1,y_2) \Big) =V_1(a_1,b_1|x_1,y_1)V_2(a_2,b_2|x_2,y_2)
\end{equation}
\end{defn}
In words, Alice and Bob play the two games independently in parallel and they win if and only if they win both of them simultaneously.

\begin{defn}
    Let $G=(X,Y,A,B, \pi,V)$ be a nonlocal game. The game $G^{\times n}$ is called the \textit{$n$-fold parallel repetition of $G$.}
\end{defn}

\subsection{Classical strategies for the $2m$-fold parallel repetition}
The following strategy was already described by Feige \cite{Feige} for the game $G_F^{\times 2}$. Receiving question $(x_1,x_2)$ and $(y_1,y_2)$, Alice and Bob answer
\begin{align*}
    a_1&=\perp,  &&b_1=y_2,\\
    a_2&=x_1,  &&b_2=\perp,
\end{align*}
respectively. We see that the players win if and only if $x_1=y_2$, which happens with probability $\frac{1}{2}$. This may be surprising, since the classical winning probability of the $2$-fold parallel repetition does not decrease when compared with the original game $G_F$.

For the $2m$-fold parallel repetition $G_F^{\times 2m}$, we obtain a classical strategy with winning probability $\frac{1}{2^m}=\frac{1}{2^{n/2}}$ by repeating the strategy above $m$ times. More precisely, given $(x_1,\dots, x_{2m})$ and $(y_1,\dots, y_{2m})$, the players answer $(a_1,\dots, a_{2m})$ and $(b_1,\dots, b_{2m})$ with
\begin{align*}
    a_{2i-1}&=\perp, && b_{2i-1}=y_{2i},\\
    a_{2i}&=x_{2i-1}, && b_{2i}=\perp,
\end{align*}
for $1\leq i\leq m$. We will see in the next section, that this probability coincides with the upper bound on the non-signalling value of $G_F^{\times 2m}$.

\subsection{The non-signalling upper bound for $2m$-fold parallel repetition}
 Let $n=2m$. Consider a set $I\subseteq [n]:=\{1,\dots, n\}$ and define $a_I(\mathbf{y})\in \{0,1, \perp\}^n, b_I(\mathbf{x})\in \{0,1, \perp\}^n$ for each $\mathbf{x}$, $\mathbf{y}\in \{0,1\}^{n}$ by
 \begin{align*}
 (a_I(\mathbf{y}))_i =\begin{cases} \perp, i\in I \\ y_i, i\notin I \end{cases} \quad \text{and} \quad (b_I(\mathbf{x}))_i =\begin{cases} \perp, i\notin I \\ x_i, i\in I. \end{cases}
 \end{align*}
 Then, we can write 
 \begin{align*}
 \omega(G_F^{\times n}, p)=\frac{1}{2^{4m}}\sum_{I\subseteq [n]}\sum_{\mathbf{x}, \mathbf{y}}p( a_I(\mathbf{y}), b_I(\mathbf{x})|\mathbf{x}, \mathbf{y}) 
 \end{align*}
 for the winning probability of the $2m$-fold parallel repetition. The goal of this subsection is to show $\omega(G_F^{\times n}, p)\leq \frac{1}{2^m}=\frac{1}{2^{n/2}}$, which is equivalent to proving 
  \begin{align*}
 \sum_{I\subseteq [n]}\sum_{\mathbf{x}, \mathbf{y}}p( a_I(\mathbf{y}), b_I(\mathbf{x})|\mathbf{x}, \mathbf{y})\leq 2^{3m}.
 \end{align*}

For this, we first show the following lemma. We define $\mathbf{x}\vert_I\in \{0,1\}^{|I|}$ to be the restriction of $\mathbf{x}$ to a subset $I\subseteq [n]$. Furthermore, we write $|I|$ for the cardinality of a set $I$ and denote the complement by $I^c=[n]\setminus I$. 
 
\begin{lem}\label{lem:evenp1}
  Let $J\subseteq \{0,1\}^{|I^c|}$. It holds
 \begin{align*}
 \sum_{\mathbf{x}\vert_{I^c}\in J}\sum_{\mathbf{x}\vert_I\in \{0,1\}^{|I|}}p( a_I(\mathbf{y}), b_I(\mathbf{x})|\mathbf{x}, \mathbf{y})\leq  |J|\sum_{\mathbf{x}\vert_I\in \{0,1\}^{|I|}}\sum_{\mathbf{a}}p( \mathbf{a}, b_I(\mathbf{x})|\mathbf{z}, \mathbf{y})
 \end{align*}
 for any $\mathbf{z}\in \{0,1\}^n$. Similarly, for $K\subseteq \{0,1\}^{|I|}$, it holds
 \begin{align*}
 \sum_{\mathbf{y}\vert_{I}\in K} \sum_{\mathbf{y}\vert_{I}\in \{0,1\}^{|I^c|}}p( a_I(\mathbf{y}), b_I(\mathbf{x})|\mathbf{x}, \mathbf{y})\leq |K| \sum_{\mathbf{y}\vert_{I^c}\in \{0,1\}^{|I^c|}}\sum_{\mathbf{b}}p( a_I(\mathbf{y}),\mathbf{b}|\mathbf{x}, \mathbf{z})
 \end{align*}
 for any $\mathbf{z}\in \{0,1\}^n$.
 \end{lem}
 
\begin{proof}
We prove the first inequality, the second follows similarly. For fixed $\mathbf{x}\vert_{I^c}$, we have  
\begin{align}\label{eq:evenparallel}
\sum_{\mathbf{x}\vert_I\in \{0,1\}^{|I|}}p( a_I(\mathbf{y}), b_I(\mathbf{x})|\mathbf{x}, \mathbf{y})\leq \sum_{\mathbf{x}\vert_I\in \{0,1\}^{|I|}}\sum_{\mathbf{a}}p( \mathbf{a}, b_I(\mathbf{x})|\mathbf{z}, \mathbf{y})
 \end{align}
 because of $p( a_I(\mathbf{y}), b_I(\mathbf{x})|\mathbf{x}, \mathbf{y})\leq \sum_{\mathbf{a}}p( \mathbf{a}, b_I(\mathbf{x})|\mathbf{x}, \mathbf{y})$ and the non-signalling property of $p$. Since $b_I(\mathbf{v})=b_I(\mathbf{w})$ if and only if $\mathbf{v}\vert_I=\mathbf{w}\vert_I$, we see that the right hand side of \Cref{eq:evenparallel} is the same for each $\mathbf{x}\vert_{I^c}\in J$. This yields the assertion.
 \end{proof}

The next lemma is concerned with sets that have exactly $m$ elements. 

  \begin{lem}\label{lem:evenp2}
Consider a subset $I\subseteq [n]$ with $m$ elements. Then
 \begin{align*}
 \sum_{\mathbf{x}, \mathbf{y}} p( a_I(\mathbf{y}), b_I(\mathbf{x})|\mathbf{x}, \mathbf{y})\leq 
 2^{m-1}\Bigg(\sum_{\mathbf{y},\mathbf{a}}&\sum_{\mathbf{x}\vert_I\in \{0,1\}^{|I|}}p( \mathbf{a}, b_I(\mathbf{x})|\mathbf{z}, \mathbf{y})\\
 &+\sum_{\mathbf{x},\mathbf{b}}\sum_{\mathbf{y}\vert_{I^c}\in \{0,1\}^{|I^c|}}p( a_I(\mathbf{y}),\mathbf{b}|\mathbf{x}, \mathbf{z})\Bigg)
 \end{align*}
 \end{lem}
 
 \begin{proof}
  Since $n=2m$, both sets $I$ and $I^c$ have $m$ elements. Therefore, we can choose pairs $(i_k, j_k)$ for $1\leq k \leq m$ with $i_k\in I, j_k\in I^c$ such that $\{i_k, j_k|1\leq k \leq m\}=[n]$. It holds
 \begin{align}
  \sum_{\mathbf{x}, \mathbf{y}} p( a_I(\mathbf{y}), b_I(\mathbf{x})|\mathbf{x}, \mathbf{y})&=\sum_{l\in \{0,1\}^m} \sum_{\mathbf{x}, \mathbf{y}}\sum_{x_{i_k}\oplus y_{j_k}=l_k \text{ for all } k} p( a_I(\mathbf{y}), b_I(\mathbf{x})|\mathbf{x}, \mathbf{y})\nonumber\\
  &=\sum_{\substack{l\in \{0,1\}^m\\\sum_kl_k=0 \text{ mod } 2}} \sum_{\mathbf{x}, \mathbf{y}}\sum_{x_{i_k}\oplus y_{j_k}=l_k \text{ for all } k} p( a_I(\mathbf{y}), b_I(\mathbf{x})|\mathbf{x}, \mathbf{y})\nonumber\\
  &\quad+\sum_{\substack{l\in \{0,1\}^m\\\sum_kl_k=1 \text{ mod } 2}} \sum_{\mathbf{x}, \mathbf{y}}\sum_{x_{i_k}\oplus y_{j_k}=l_k \text{ for all } k} p( a_I(\mathbf{y}), b_I(\mathbf{x})|\mathbf{x}, \mathbf{y}).\label{eq:evenpa}
 \end{align}
 Now, \Cref{lem:evenp1} yields 
 \begin{align*}
 \sum_{\substack{l\in \{0,1\}^m\\\sum_kl_k=0 \text{ mod } 2}} \sum_{\mathbf{x}, \mathbf{y}}\sum_{\substack{x_{i_k}\oplus y_{j_k}=l_k\\ \text{ for all } k}} p( a_I(\mathbf{y}), b_I(\mathbf{x})|\mathbf{x}, \mathbf{y})\leq 2^{m-1}\sum_{\mathbf{y}, \mathbf{a}}\sum_{\mathbf{x}\vert_I\in \{0,1\}^{|I|}}p( \mathbf{a}, b_I(\mathbf{x})|\mathbf{z}, \mathbf{y})
 \end{align*}
 since the set $\{l\in \{0,1\}^m|\sum_kl_k=0 \text{ mod } 2\}$ has $2^{m-1}$ elements. In the same way, we obtain
  \begin{align*}
 \sum_{\substack{l\in \{0,1\}^m\\\sum_kl_k=1 \text{ mod } 2}} \sum_{\mathbf{x}, \mathbf{y}}\sum_{\substack{x_{i_k}\oplus y_{j_k}=l_k\\ \text{ for all } k}} p( a_I(\mathbf{y}), b_I(\mathbf{x})|\mathbf{x}, \mathbf{y})\leq 2^{m-1}\sum_{\mathbf{x},\mathbf{b}}\sum_{\mathbf{y}\vert_{I^c}\in \{0,1\}^{|I^c|}}p( a_I(\mathbf{y}),\mathbf{b}|\mathbf{x}, \mathbf{z}).
 \end{align*}
 Together with \Cref{eq:evenpa} this yields the result. 
 \end{proof}

Now, we can prove the main theorem of this subsection. 
 
 \begin{thm}
 It holds $\omega_{ns}(G_F^{\times n})\leq \frac{1}{2^{n/2}}$ for $n$ even.
 \end{thm}
 
 \begin{proof}
 We will prove
   \begin{align*}
 2^{4m}\omega(G_F^{\times n}, p)=\sum_{I\subseteq [n]}\sum_{\mathbf{x}, \mathbf{y}}p( a_I(\mathbf{y}), b_I(\mathbf{x})|\mathbf{x}, \mathbf{y})\leq 2^{3m}
 \end{align*}
 for any nonsignalling strategy $p$. From \Cref{lem:evenp1}, we obtain for $|I|>m$ that 
\begin{align}
\sum_{\mathbf{x}, \mathbf{y}}p( a_I(\mathbf{y}), b_I(\mathbf{x})|\mathbf{x}, \mathbf{y})&= \sum_\mathbf{y}\sum_{\mathbf{x}\vert_{I^c}\in \{0,1\}^{|I|^c}}\sum_{\mathbf{x}\vert_I\in \{0,1\}^{|I|}}p( a_I(\mathbf{y}), b_I(\mathbf{x})|\mathbf{x}, \mathbf{y})\nonumber\\
&\leq  2^{m-1}\sum_{\mathbf{y}}\sum_{\mathbf{x}\vert_I\in \{0,1\}^{|I|}}\sum_{\mathbf{a}}p( \mathbf{a}, b_I(\mathbf{x})|\mathbf{z}, \mathbf{y}),\label{eq:evenp1}
\end{align} 
since $|\{0,1\}^{|I|^c}|\leq 2^{m-1}$. Similarly, we get 
\begin{align}
\sum_{\mathbf{x}, \mathbf{y}}p( a_I(\mathbf{y}), b_I(\mathbf{x})|\mathbf{x}, \mathbf{y})\leq 2^{m-1}\sum_{\mathbf{x}}\sum_{\mathbf{y}\vert_{I^c}\in \{0,1\}^{|I^c|}}\sum_{\mathbf{b}}p( a_I(\mathbf{y}),\mathbf{b}|\mathbf{x}, \mathbf{z})\label{eq:evenp2}
\end{align}
for $|I|<m$. 
Together with \Cref{lem:evenp2}, \Cref{eq:evenp1,eq:evenp2} yield
   \begin{align}
 \sum_{I\subseteq [n]}\sum_{\mathbf{x}, \mathbf{y}}p( a_I(\mathbf{y}), b_I(\mathbf{x})|\mathbf{x}, \mathbf{y})\leq 2^{m-1}\Bigg(&\sum_{I, |I|\geq m} \sum_{\mathbf{y}}\sum_{\mathbf{x}\vert_I\in \{0,1\}^{|I|}}\sum_{\mathbf{a}}p( \mathbf{a}, b_I(\mathbf{x})|\mathbf{z}, \mathbf{y})\nonumber\\&+\sum_{I, |I|\leq m}\sum_{\mathbf{x}}\sum_{\mathbf{y}\vert_{I^c}\in \{0,1\}^{|I^c|}}\sum_{\mathbf{b}}p( a_I(\mathbf{y}),\mathbf{b}|\mathbf{x}, \mathbf{z})\Bigg) \label{eq:evenp3}
 \end{align}
 Note that we have 
 \begin{align*}
\sum_{I, |I|\geq m}\sum_{\mathbf{x}\vert_I\in \{0,1\}^{|I|}}\sum_{\mathbf{a}}p( \mathbf{a}, b_I(\mathbf{x})|\mathbf{z}, \mathbf{y})\leq \sum_{\mathbf{a}, \mathbf{b}} p( \mathbf{a}, \mathbf{b})|\mathbf{z}, \mathbf{y})=1
 \end{align*}
 since $b_{I_1}(\mathbf{v})\neq b_{I_2}(\mathbf{w})$ for $I_1\neq I_2$ and any $\mathbf{v}, \mathbf{w}$. Similarly, we obtain 
 \begin{align*}
 \sum_{I, |I|\leq m}\sum_{\mathbf{x}}\sum_{\mathbf{y}\vert_{I^c}\in \{0,1\}^{|I^c|}}\sum_{\mathbf{b}}p( a_I(\mathbf{y}),\mathbf{b}|\mathbf{x}, \mathbf{z})\leq 1.
 \end{align*}
From \Cref{eq:evenp3}, we then deduce 
\begin{align*}
 \sum_{I\subseteq [n]}\sum_{\mathbf{x}, \mathbf{y}}p( a_I(\mathbf{y}), b_I(\mathbf{x})|\mathbf{x}, \mathbf{y})\leq 2^{m-1}(2^{2m}+2^{2m})=2^{3m}
\end{align*}
since there are $2^{2m}$ different $\mathbf{x}$ and $\mathbf{y}$, respectively. This finishes the proof. 
 \end{proof}

 \subsection{Values for the $3$-fold parallel repetition}
Now, we will look at the case of $3$-fold parallel repetition. We will start with describing a classical strategy with winning probability $\frac{5}{16}$. We note that this strategy performs better that the quantum strategy consisting of the optimal quantum strategy for one game and the optimal classical strategy for the $2$-fold parallel repetition, since $\frac{9}{16}\cdot\frac{1}{2}=\frac{9}{32}<\frac{5}{16}$.
Alice's and Bob's strategies are of the form $\left(a_1, a_2 , a_3 \right)$, $\left(b_1 , b_2 , b_3 \right)$
where
$$
a_i : \bit^3 \rightarrow \{0,1,\perp \}, \text{ } b_i : \bit^3 \rightarrow \{0,1,\perp \}, \quad i=1,2,3
$$
We will show that the success probability using the following strategy is $5/16$. Given $(x_1,x_2,x_3)$ and $(y_1,y_2,y_3)$, Alice and Bob answer with
\begin{align*}
    a_1&=\perp, &&b_1=y_2\wedge y_3,\\
    a_2&=x_1 \vee x_3, &&b_2=\perp,\\
    a_3&=\begin{cases}
        \perp \text{ if }x_1=0\\
        1 \hspace{0.2cm}\text{ if }x_1=1
    \end{cases}\hspace{-0.4cm},
    &&b_3=\begin{cases}
        y_2 \hspace{0.2cm}\text{ if }y_2\wedge y_3=0\\
        \perp \hspace{0.15cm}\text{ if }y_2\wedge y_3=1
    \end{cases}\hspace{-0.4cm},
\end{align*}
respectively. Let us denote the event that Alice and Bob win round $i=1,2,3$ by $W_i$. Then
\begin{align*}
\Pr[W_1]&=\Pr[x_1=y_2\wedge y_3 ]=\frac{1}{2},\\
\Pr[W_2|W_1]&=\Pr[y_2=x_1 \vee x_3|x_1=y_2\wedge y_3]= \frac{5}{8}.
\end{align*}
To see the latter, note that
\begin{align*}
\Pr[y_2=x_1 \vee x_3|x_1=y_2\wedge y_3]&=\Pr[y_2=(y_2\wedge y_3)\vee x_3]\\
&=\frac{1}{2}\Pr[y_3\vee x_3=1]+\frac{1}{2}\Pr[x_3=0]\\
&=\frac{1}{2}\cdot \frac{3}{4}+\frac{1}{4}\\
&=\frac{5}{8},
\end{align*}
where we used that the questions are sampled independently. Now, since
\begin{align}
    \Pr[W_1,W_2,W_3] &= \Pr[W_2|W_1,W_3] \Pr[W_1|W_3]\Pr[W_3]\nonumber\\
        &=\frac{5}{16} \Pr[W_2|W_1,W_3],\label{eq:classicaln3strat}
\end{align}
the goal is to show that $\Pr[W_2|W_1,W_3] =1$. Assume
\begin{align*}
    &x_1=y_2\wedge y_3, &&y_2=x_1\vee x_3.
\end{align*}
If $x_1=y_2\wedge y_3=0$, the players win round $3$ if and only if $y_2=x_3$. But this is true because $y_2=x_1\vee x_3=x_3$ for $x_1=0$. If $x_1=y_2\wedge y_3=1$, the players win if and only if $y_3=1$, which is true since $y_2\wedge y_3=1$ implies $y_3=1$. This shows $\Pr[W_2|W_1,W_3] =1$ and by \Cref{eq:classicaln3strat}, the strategy has winning probability $\frac{5}{16}$. By exhaustive search through all deterministic strategies using a computer, we have confirmed that this strategy is indeed optimal.

While we see that the previous classical strategy performs better than the quantum strategy consisting of the optimal quantum strategy for one game and the optimal classical strategy for the $2$-fold parallel repetition, there is no non-signalling strategy that outperforms the one where we use optimal non-signalling strategy for one game and the optimal classical strategy for the $2$-fold parallel repetition. Therefore, the non-signalling value of the $3$-fold parallel repetition is $\frac{1}{3}$. The following is an optimal strategy achieving this winning probability.

\begin{align*}
    p(\mathbf{a},\mathbf{b}|\mathbf{x},\mathbf{y})=\begin{cases}
        \frac{1}{3} \quad\text{ if }[(a_1,b_1)=(\perp,x_1) \text{ or } (a_1,b_1)=(y_1,\perp)] \\\hspace{1cm}\text{ and }(a_2,a_3)=(\perp,x_2), (b_2,b_3)=(y_3,\perp),\\
        \frac{1}{3} \quad\text{ if } a_1\neq y_1, b_1\neq x_1, \text{ and } a_1,b_1 \in \{0,1\}\\
        \hspace{1cm}\text{ and }(a_2,a_3)=(\perp,x_2), (b_2,b_3)=(y_3,\perp),\\
        0 \quad \text{ otherwise.}
    \end{cases}
\end{align*}

We also obtain $\omega_{ns}(G^{\times 3})=\frac{1}{3}$, which can be verified in a similar fashion as done for the cases $n=1$ and $n=2m$ in this article.

\section{Feige's game as synchronous game}

A nonlocal game is called \emph{synchronous} if the question and answer sets of the players coincide, and additionally they have to return the same answer if they receive the same question in order to win. More formally, it holds $X=Y$, $A=B$ and $V(a,b \,|\, x,x)=0$ for $a\neq b$. There is the following synchronous version of the Feige game, which we denote by $G_F^s$:
\begin{align*}
    X= Y=\{0,1 \}, \quad A=B=\{A0,A1,B0,B1 \}, \quad \pi(x,y)=\tfrac{1}{4}
\end{align*}

\begin{align*}
    V(a,b,x,y)=\begin{cases}
        1 \quad (a,b)=(Ax,Ax) \text{ or } (a,b)=(By,By)
        \\
        0 \quad \text{otherwise}.
    \end{cases}
\end{align*}

Note that Alice should never answer $A1$, when she receives $0$, since the players always lose in this case. The same is true for Bob, which shows that to achieve maximum winning probability, the players will restrict to three answers as in the original Feige game. The difference is that $G_F^s$ is a synchronous game, whereas $G_F$ is not synchronous.

We define a strategy $S$ to be \emph{synchronous} if $p(a,b\,|\, x,x)=0$ for $a\neq b$, that is, the players never give different answers when receiving the same question. We can now define the \emph{synchronous value} $\omega_q^s(G)$ of a nonlocal game $G$ to be the supremum over all synchronous quantum strategies, i.e.
\begin{align*}
\omega_q^s(G)=\mathrm{sup}\{\omega(S,G)\,|\,S \text{ synchronous quantum strategy}\}.
\end{align*}
The goal of this section is to prove that
\begin{align*}
\omega_q^s(G_F^s)=\frac{1}{2}<\frac{9}{16}=\omega_q(G_F^s).
\end{align*}
Therefore, we have a synchronous game for which the synchronous value is equal to the classical value, but the game has quantum advantage. Such games are known from \cite{HMNPR}, where they considered certain graph coloring games. Note that $\omega_q(G_F^s)\geq\frac{9}{16}$ by using the optimal quantum strategy for the Feige game for the synchronous version, and $\omega_q(G_F^s)\leq\frac{9}{16}$ follows from the NPA hierarchy at level $1+AB$, in the same way as for $G_F$. In the following, we will show $\omega_q^s(G_F^s)=\frac{1}{2}$. Recall the following lemma from \cite[Corollary 3.6 a)]{MPS}.

\begin{lem}\label{lem:MPS}
Let $S=(\ket{\psi}, \{P^x_a\}, \{Q^y_b\})$ be a synchronous quantum strategy. Then
\begin{align*}
P^x_a\otimes \mathrm{Id}\ket{\psi}= \mathrm{Id}\otimes Q^x_a\ket{\psi}
\end{align*}
for any $x,a$.
\end{lem}

With this we can prove our next result.

\begin{prop}
It holds $\omega_q^s(G_F^s)=\frac{1}{2}$.
\end{prop}

\begin{proof}
Consider a synchronous quantum strategy $S=(\ket{\psi}, \{P^x_a\}, \{Q^y_b\})$ for $G_F^s$. The winning probability of $S$ is given by
\begin{align*}
\omega_q(S, G_F^s)=\frac{1}{4}\sum_{x,y\in\{0,1\}}\bra{\psi} P^x_{Ax}\otimes Q^y_{Ax}+P^x_{By}\otimes Q^y_{By}\ket{\psi}.
\end{align*}
Since $\mathrm{Id}-P^x_{Ax}$ and $\mathrm{Id}-Q^y_{By}$ are positive operators, we get $\bra{\psi} P^x_{Ax}\otimes Q^y_{Ax}\ket{\psi}\leq \bra{\psi} \mathrm{Id}\otimes Q^y_{Ax}\ket{\psi}$ and $\bra{\psi} P^x_{By}\otimes Q^y_{By}\ket{\psi}\leq \bra{\psi} P^x_{By}\otimes \mathrm{Id}\ket{\psi}$. Therefore
\begin{align*}
    \omega_q(S, G_F^s)&\leq\frac{1}{4}\sum_{x,y\in\{0,1\}}\bra{\psi} \mathrm{Id}\otimes Q^y_{Ax}+P^x_{By}\otimes \mathrm{Id}\ket{\psi}\\
    &=\frac{1}{4}\sum_{x,y\in\{0,1\}}\bra{\psi} (P^y_{Ax}+P^x_{By})\otimes \mathrm{Id}\ket{\psi},
\end{align*}
where we used \Cref{lem:MPS} with $Q^y_{Ax}$ in the second step. We conclude
\begin{align*}
    \omega_q(S, G_F^s)&\leq\frac{1}{4} \bra{\psi} (P^0_{A0}+P^0_{B0}+P^0_{A1}+P^0_{B1})+(P^1_{A0}+P^1_{B0}+P^1_{A1}+P^1_{B1})\otimes \mathrm{Id}\ket{\psi}
    =\frac{1}{2}.
\end{align*}
\end{proof}

\section*{Acknowledgments}
We thank Laura Man\v{c}inska for helpful discussions about "or"-games. We acknowledge support by the European Research Council (ERC Grant SYMOPTIC, 101040907), the German Federal Ministry of Research, Technology and Space (QuSol, 13N17173), and the German Research Foundation under Germany's Excellence Strategy (EXC 2092 CASA, 390781972). Y.Z.~is supported by VILLUM FONDEN via QMATH Centre of Excellence grant number 10059 and Villum Young Investigator grant number 37532. S.A.L.S is supported by the European Union via a ERC grant (QInteract, Grant No 101078107) as well as VILLUM FONDEN via Villum Young Investigator grant (No 37532).

\clearpage

\begin{appendix}
    \section{Algebraic structure of Feige's game from sums of square decompositions} \label{appedix: SOS decomp.}

The vector of polynomials
$$
F=F(p_a^x, q_b^y)=\left(f_1(p_a^x, q_b^y), ... ,f_{32}(p_a^x, q_b^y)\right)^T
$$ 
in the proof of \cref{lem:pair} is given by
\begin{equation}\nonumber
    f_1(p_a^x, q_b^y)  = p_0^0 p_1^1 p_0^0-\frac{1}{16}p_0^0
\end{equation}
\begin{equation} \nonumber
    f_2(p_a^x, q_b^y)  = p_1^0 p_1^1 p_1^0-\frac{9 }{16}p_1^0
\end{equation}
\begin{equation} \nonumber
    f_3(p_a^x, q_b^y)  = p_0^1 p_1^0 p_0^1-\frac{1}{16}p_0^1
\end{equation}
\begin{equation} \nonumber
    f_4(p_a^x, q_b^y)  = p_1^1 p_1^0 p_1^1-\frac{9 }{16}p_1^1
\end{equation}
\begin{equation} \nonumber
    f_5(p_a^x, q_b^y)  = p_0^0 p_0^1 p_0^0-\frac{9 }{16}p_0^0
\end{equation}
\begin{equation} \nonumber
    f_6(p_a^x, q_b^y)  = p_0^1 p_0^0 p_0^1-\frac{9 }{16}p_0^1
\end{equation}
\begin{equation} \nonumber
    f_{7} (p_a^x, q_b^y)  = p_1^0 p_0^1 p_1^0-\frac{1}{16}p_1^0
\end{equation}

\begin{equation} \nonumber
f_8(p_a^x , q_b^y)=p_1^1 p_0^0 p_1^1-\frac{1}{16}p_1^1
\end{equation}

\begin{equation} \nonumber
f_9(p_a^x , q_b^y)= \frac{3}{4} p_0^1 p_1^0+\frac{3}{4} p_1^0 p_0^1+\frac{5}{4} p_1^0 p_1^1+\frac{5}{4} p_1^1 p_1^0+\frac{5 }{16}p_0^0-\frac{3 }{8}p_0^1-\frac{9 }{16}p_1^0-\frac{15 }{8}p_1^1+q_0^0
\end{equation}

\begin{equation} \nonumber
f_{10}(p_a^x , q_b^y)= \frac{3}{4} p_0^1 p_1^0+\frac{3}{4} p_1^0 p_0^1+\frac{5}{4} p_1^0 p_1^1+\frac{5}{4} p_1^1 p_1^0-\frac{3 }{16}p_0^0+\frac{1}{8}p_0^1-\frac{33 }{16}p_1^0-\frac{3 }{8}p_1^1+q_1^0
\end{equation}

\begin{equation} \nonumber
f_{11}(p_a^x , q_b^y)= \frac{3}{4} p_0^1 p_1^0+\frac{3}{4} p_1^0 p_0^1+\frac{5}{4} p_1^0 p_1^1+\frac{5}{4} p_1^1 p_1^0+\frac{9 }{16}p_0^0 - \frac{9 }{8}p_0^1 - \frac{13 }{16} p_1^0 -\frac{9 }{8} p_1^1 +q_0^1
\end{equation}

\begin{equation} \nonumber
f_{12}(p_a^x , q_b^y)= \frac{3}{4} p_0^1 p_1^0+\frac{3}{4} p_1^0 p_0^1+\frac{5}{4} p_1^0 p_1^1+\frac{5}{4} p_1^1 p_1^0-\frac{15 }{16}p_0^0 + \frac{3 }{8} p_0^1 - \frac{21 }{16} p_1^0 - \frac{5 }{8}p_1^1 + q_1^1
\end{equation}

\begin{equation} \nonumber
f_{13}(p_a^x , q_b^y)= I+ \frac{2}{3} p_0^1 p_1^0+\frac{2}{3} p_1^0 p_0^1+2 p_1^0 p_1^1+2 p_1^1 p_1^0-\frac{1}{2}p_0^0-\frac{2 }{3}p_0^1 -\frac{13 }{6}p_1^0 - 2 p_1^1
\end{equation}

\begin{equation} \nonumber
f_{14}(p_a^x , q_b^y)= p_0^0  q_0^0-\frac{3}{8} p_0^1 p_1^0-\frac{3}{4} p_1^0 p_0^1+\frac{3}{4} p_1^1 p_0^0+\frac{1}{8} p_1^1 p_1^0-\frac{5 }{32}p_0^0+\frac{3 }{16}p_0^1+\frac{3 }{32}p_1^0-\frac{3 }{16}p_1^1
\end{equation}

\begin{equation} \nonumber
f_{15}(p_a^x , q_b^y)= p_0^0  q_0^1+\frac{3}{8} p_0^1 p_1^0+\frac{9}{4} p_1^0 p_0^1-\frac{9}{4} p_1^1 p_0^0-\frac{5}{8} p_1^1 p_1^0+\frac{9 }{32}p_0^0-\frac{9 }{16}p_0^1-\frac{3 }{32}p_1^0+\frac{9 }{16}p_1^1
\end{equation}

\begin{equation} \nonumber
f_{16}(p_a^x , q_b^y)= p_0^1  q_0^0+\frac{9}{8} p_0^1 p_1^0+\frac{15}{8} p_1^0 p_0^1+\frac{5}{8} p_1^0 p_1^1-\frac{15}{8} p_1^1 p_0^0+\frac{15 p_0^0}{32}-\frac{9 p_0^1}{16}-\frac{15 p_1^0}{32}
\end{equation}

\begin{equation} \nonumber
f_{17}(p_a^x , q_b^y)= p_0^1  q_0^1+\frac{7}{8} p_0^1 p_1^0+\frac{21}{8} p_1^0 p_0^1+\frac{7}{8} p_1^0 p_1^1-\frac{21}{8} p_1^1 p_0^0+\frac{21 }{32}p_0^0-\frac{21 }{16}p_0^1-\frac{21 }{32}p_1^0
\end{equation}

\begin{equation} \nonumber
f_{18}(p_a^x , q_b^y)= p_0^0  q_1^0+\frac{9}{8} p_0^1 p_1^0-\frac{3}{4} p_1^0 p_0^1+\frac{3}{4} p_1^1 p_0^0+\frac{5}{8} p_1^1 p_1^0-\frac{9 }{32}p_0^0+\frac{3 }{16}p_0^1-\frac{9 }{32}p_1^0-\frac{3 }{16}p_1^1
\end{equation}

\begin{equation} \nonumber
f_{19}(p_a^x , q_b^y)= p_0^0  q_1^1+\frac{7}{8} p_0^1 p_1^0-\frac{7}{4} p_1^0 p_0^1+\frac{7}{4} p_1^1 p_0^0+\frac{7}{8} p_1^1 p_1^0-\frac{35 }{32}p_0^0+\frac{7 }{16}p_0^1-\frac{7 }{32}p_1^0-\frac{7 }{16}p_1^1
\end{equation}

\begin{equation} \nonumber
f_{20}(p_a^x , q_b^y)= p_0^1  q_1^0-\frac{3}{8} p_0^1 p_1^0+\frac{3}{8} p_1^0 p_0^1+\frac{1}{8} p_1^0 p_1^1-\frac{3}{8} p_1^1 p_0^0+\frac{3 }{32}p_0^0-\frac{1}{16}p_0^1-\frac{3 }{32}p_1^0
\end{equation}

\begin{equation} \nonumber
f_{21}(p_a^x , q_b^y)= p_0^1  q_1^1+\frac{3}{8} p_0^1 p_1^0-\frac{15}{8} p_1^0 p_0^1-\frac{5}{8} p_1^0 p_1^1+\frac{15}{8} p_1^1 p_0^0-\frac{15 }{32}p_0^0+\frac{3 }{16}p_0^1+\frac{15 }{32}p_1^0
\end{equation}

\begin{equation} \nonumber
f_{22}(p_a^x , q_b^y)= p_1^0  q_0^0+\frac{3}{8} p_1^0 p_0^1-\frac{5}{8} p_1^0 p_1^1+\frac{3 }{16}p_1^0
\end{equation}

\begin{equation} \nonumber
f_{23}(p_a^x , q_b^y)= p_1^0  q_0^1-\frac{3}{8} p_1^0 p_0^1+\frac{1}{8} p_1^0 p_1^1-\frac{1}{16}p_1^0
\end{equation}

\begin{equation} \nonumber
f_{24}(p_a^x , q_b^y)= p_1^1  q_0^0+\frac{7}{8} p_1^1 p_0^0+\frac{7}{8} p_1^1 p_1^0-\frac{21 }{16}p_1^1
\end{equation}

\begin{equation} \nonumber
f_{25}(p_a^x , q_b^y)= p_1^1  q_0^1+\frac{9}{8} p_1^1 p_0^0+\frac{5}{8} p_1^1 p_1^0-\frac{9 }{16}p_1^1
\end{equation}

\begin{equation} \nonumber
f_{26}(p_a^x , q_b^y)= p_1^0  q_1^0+\frac{7}{8} p_1^0 p_0^1+\frac{7}{8} p_1^0 p_1^1-\frac{21 }{16}p_1^0
\end{equation}

\begin{equation} \nonumber
f_{27}(p_a^x , q_b^y)= p_1^0  q_1^1+\frac{9}{8} p_1^0 p_0^1+\frac{5}{8} p_1^0 p_1^1-\frac{9 }{16}p_1^0
\end{equation}

\begin{equation} \nonumber
f_{28}(p_a^x , q_b^y)= p_1^1  q_1^0+\frac{3}{8} p_1^1 p_0^0-\frac{5}{8} p_1^1 p_1^0+\frac{3 }{16}p_1^1
\end{equation}

\begin{equation} \nonumber
f_{29}(p_a^x , q_b^y)= p_1^1  q_1^1-\frac{3}{8} p_1^1 p_0^0+\frac{1}{8} p_1^1 p_1^0-\frac{1}{16}p_1^1
\end{equation}

\begin{equation} \nonumber
f_{30}(p_a^x , q_b^y)= p_0^0 p_0^1+3 p_1^0 p_0^1-3 p_1^1 p_0^0-p_1^1 p_1^0-\frac{3 }{4}p_0^1+\frac{3 }{4}p_1^1
\end{equation}

\begin{equation} \nonumber
f_{31}(p_a^x , q_b^y)= p_0^1 p_0^0-3 p_1^0 p_0^1-p_1^0 p_1^1+3 p_1^1 p_0^0-\frac{3 }{4}p_0^0+\frac{3 }{4}p_1^0
\end{equation}

\begin{equation} \nonumber
f_{32}(p_a^x , q_b^y)= p_0^0 p_1^1-p_0^1 p_1^0-p_1^0 p_0^1+p_1^1 p_0^0-\frac{1}{4}p_0^0+\frac{1}{4}p_0^1+\frac{1}{4}p_1^0-\frac{1}{4}p_1^1
\end{equation}

\section{A generalization of mutually unbiased measurements}\label{sect:generalMUM}
For $n\in \N$ and $\nu=(\nu_{ij})_{i,j=1}^n$ a real symmetric $n\times n$ matrix with $\nu_{ij}>0$ satisfying $\sum_{j} \nu_{ij}=1$, let $\cA_\nu$ be the universal $C^*$-algebra generated by $E_1,...,E_n,F_1, ... ,F_n$ satisfying:
\begin{itemize}
    \item[i)] $E_i = E_i^*$, $F_i=F_i^*$ for all $i=1,...,n$
    \item[ii)] $\sum_{i=1}^n E_i =\sum_{i=1}^n F_i = I$
    \item[iii)] $E_iF_jE_i=\nu_{ij} E_i$, $F_iE_jF_i=\nu_{ij}F_i$ for all $i,j=1,...,n$
\end{itemize}
Observe that by summing over $j$ in iii) and using i) and ii), one easily sees that the generators consist of two sets of orthogonal projections.

\begin{defn}
    We call a representation $\rho_\nu: \cA_\nu \rightarrow B(H)$ a \textit{$\nu$-biased measurement}.
\end{defn}
This definition becomes clear when considering the following example
\begin{exa}
    A pair of $n$-outcome \textit{mutually unbiased measurements (MUMs)} as defined in \cite{MUM_Tavakoli_etal} constitute a representation $\cA_{d^{-1}1_n}$ where $1_n$ is the $n\times n$ matrix of all ones.
\end{exa}
Hence $\nu$-biased bases generalize the concept of MUMs (see e.g. \cite{MUM_Tavakoli_etal}). It turns out we can generalize results from \cite{MUM_Tavakoli_etal} as follows.
\begin{prop}
Let $\rho_\nu: \cA_\nu \rightarrow B(H)$ ($H$ separable) be a $\nu$-biased measurement, $P_a:=\rho(E_a)$ and $Q_b:=\rho(F_b)$. Then
$$
H \cong H' \otimes \C^n
$$
for some Hilbert space $H'$ and
\begin{align}
    P_a & =  I_{H'} \otimes \dyad{a}{a}
    \\
    Q_b & =  \sum_{jk} V_j^{(b)} (V_k^{(b)})^*  \otimes \dyad{j}{k}
\end{align}
where the operators $V^{(b)}_j$ satisfy
\begin{align}
    V_{j}^{(b)}(V_{j}^{(b)})^* & = \nu_{jb} I_{H'}
    \\
    V_1^{(b)} & = \sqrt{\nu_{1b}}I_{H'}
    \\
    V_j^{(1)} & = \sqrt{\nu_{1j}}I_{H'}
    \\
    \sum_k  (V_k^{(b)})^* V_{k}^{(b')} & = 0, \quad \forall b\neq b
    \\
    \sum_b V_j^{(b)} (V_k^{(b)})^* & = \delta_{jk}I_{H'}
\end{align}
\end{prop}
\begin{proof}
The operator $T_{ab}:=\tfrac{1}{\sqrt{\nu_{ab}}} P_a Q_b$ is a partial isometry with support projection $Q_b$ and range projection $P_a$ i.~e.
$$
T_{ab}T_{ab}^*=P_a, \quad T_{ab}^*T_{ab} =Q_b.
$$
This implies that either all projections are finite rank and $\Tr(P_i)=\Tr(Q_j)=:\tau$ for all $i,j=1,...,n$ or none of them are.
Notice that for finite dimensional representations we immediately have that the dimension is a multiple of $n$. Indeed,
$$
n\tau =\sum_{i=1}^n \Tr\left( P_i\right) = \Tr(I) = d.
$$
The completeness relation $\sum_i P_i =I$ and the above considerations allow us to write
$$
H\cong \bigoplus_{i=1}^n H_i \cong H'\otimes \C^n
$$
where $H_i\subseteq H$ is the subspace, $P_i$ projects onto and $H'$ is another (potentially infinite dimensional) Hilbert space. Now we write
$$
P_a = I_{H'} \otimes \dyad{a}{a}, \quad Q_b =\sum_{ij} X_{ij}^{(b)} \otimes \dyad{i}{j}
$$
for some ONB $\{ \ket a\}$ and some operators $X_{ij}^{(b)} \in B(H')$. First note that
$$
(X_{ij}^{(b)})^*=X_{ji}^{(b)}
$$
since $Q_b$ is self-adjoint. Now, by assumption,
$$
P_a Q_b P_a = X_{aa}^{(b)} \otimes \dyad{a}{a} = \nu_{ab} I_{H'} \otimes \dyad{a}{a}
$$
which implies
\begin{equation}\label{eq:rel1}
    X_{aa}^{(b)}=\nu_{ab}I_{H'}
\end{equation}
and
$$
Q_b P_a Q_b = \sum_{ij} X_{ia}^{(b)}X_{aj}^{(b)} \otimes \dyad{i}{j} =  \sum_{ij} \nu_{ab}X_{ij}^{(b)} \otimes \dyad{i}{j}$$
which implies
\begin{equation} \label{eq:rel2}
    X_{ia}^{(b)}X_{aj}^{(b)} = \nu_{ab}X_{ij}^{(b)}
\end{equation}
We will use \eqref{eq:rel1}-\eqref{eq:rel2} several times in the following. We claim now that the operator,
$$
U := \tfrac{1}{\sqrt{\nu_{11}}} \sum_{j} \tfrac{1}{\sqrt{\nu_{1j}}} X_{1j}^{(1)} \otimes \dyad{j}{j}
$$
is unitary. Indeed,
$$
UU^* =\tfrac{1}{\nu_{11}} \sum_{j}\tfrac{1}{\nu_{1j}} X_{1j}^{(1)}X_{j1}^{(1)} \otimes \dyad{j}{j} = \tfrac{1}{\nu_{11}}  \sum_{j}  X_{11}^{(1)} \otimes \dyad{j}{j}= \sum_{j} I_{H'} \otimes \dyad{j}{j}=I_{H}
$$
and similarly $U^*U =I_H$.
Then we have
$$
UP_aU^*=\tfrac{1}{\nu_{11}} \tfrac{1}{\nu_{1a}} X_{1a}^{(1)}X_{a1}^{(1)} \otimes \dyad{j}{j}=\tfrac{1}{\nu_{11}}  X_{11}^{(1)} \otimes \dyad{j}{j}=P_a
$$
and
$$
Q_b':=U Q_b U^*= \tfrac{1}{\nu_{11}}  \sum_{jk} \tfrac{1}{\sqrt{\nu_{1j} \nu_{1k} }} X_{1j}^{(1)} X_{jk}^{(b)}X_{k1}^{(1)} \otimes \dyad{j}{k}
$$
In particular we find
\begin{equation} \nonumber
    \begin{split}
        Q_1' = & \tfrac{1}{\nu_{11}}  \sum_{jk} \tfrac{1}{\sqrt{\nu_{1j} \nu_{1k} }} X_{1j}^{(1)} X_{jk}^{(1)}X_{k1}^{(1)} \otimes \dyad{j}{k} = \tfrac{1}{\nu_{11}}  \sum_{jk} \tfrac{\sqrt{\nu_{1j}}}{\sqrt{ \nu_{1k} }} X_{1k}^{(1)}X_{k1}^{(1)} \otimes \dyad{j}{k}
        \\
        = & \tfrac{1}{\nu_{11}}  \sum_{jk} \sqrt{\nu_{1j}}\sqrt{ \nu_{1k}}  X_{11}^{(1)} \otimes \dyad{j}{k} =   \sum_{jk} \sqrt{\nu_{1j}}\sqrt{ \nu_{1k}}  I_{H'} \otimes \dyad{j}{k} =I_{H'}\otimes \dyad{v}{v}
    \end{split}
\end{equation}
where $\ket v := \sum_{j} \sqrt{\nu_{1j}} \ket j$ is a unit vector.
Now we characterize the remaining $Q_b'$. Let
$$
Y_{jk}^{(b)} = \tfrac{1}{\nu_{11} \sqrt{\nu_{1j}\nu_{1k}}} X_{1j}^{(1)} X_{jk}^{(b)}X_{k1}^{(1)}
$$
so that
$$
Q_b' = \sum_{jk} Y_{jk}^{(b)} \otimes \dyad{j}{k}
$$
and notice that
\begin{equation} \label{int}
    \begin{split}
        Y_{ja}^{(b)}Y_{ak}^{(b)}= & \tfrac{1}{\nu_{11}^2 \nu_{1a} \sqrt{\nu_{1j}\nu_{1k}} } (X_{1j}^{(1)} X_{ja}^{(b)}X_{a1}^{(1)})(X_{1a}^{(1)} X_{ak}^{(b)}X_{k1}^{(1)})
        =\tfrac{1}{\nu_{11}\sqrt{\nu_{1j}\nu_{1k}}} X_{1j}^{(1)} X_{ja}^{(b)} X_{ak}^{(b)}X_{k1}^{(1)}
        \\
        = & \tfrac{\nu_{ab}}{\nu_{11}\sqrt{\nu_{1j}\nu_{1k}}} X_{1j}^{(1)} X_{jk}^{(b)} X_{k1}^{(1)} =\nu_{ab} Y_{jk}^{(b)}
    \end{split}
\end{equation}
%
Hence we can write
$$
Q_b' = \sum_{jk} V_j^{(b)} (V_k^{(b)})^*  \otimes \dyad{j}{k}
$$
where
$$
V_j^{(b)} :=\tfrac{1}{\sqrt{\nu_{1b}}} Y_{j1}^{(b)}= \tfrac{1}{\sqrt{\nu_{1b}\nu_{1j}\nu_{11}}} X_{1j}^{(1)} X_{j1}^{(b)}.
$$
so
$$
V_1^{(b)}=\tfrac{1}{\sqrt{\nu_{1b}\nu_{11}\nu_{11}}} X_{11}^{(1)} X_{11}^{(b)} = \sqrt{\nu_{1b}}I
$$
and
$$
V_j^{(1)} :=\tfrac{1}{\sqrt{\nu_{11}}} Y_{j1}^{(1)}= \tfrac{1}{\sqrt{\nu_{11}\nu_{1j}\nu_{11}}} X_{1j}^{(1)} X_{j1}^{(1)}= \sqrt{\nu_{1j}}I.
$$
Moreover, we have
\begin{equation} \nonumber
    \begin{split}
        V_{j}^{(b)}(V_{j}^{(b)})^* = & \tfrac{1}{\nu_{1b}\nu_{1j}\nu_{11}}X_{1j}^{(1)} X_{j1}^{(b)}
        X_{1j}^{(b)} X_{j1}^{(1)}=\tfrac{1}{\nu_{1j}\nu_{11}}X_{1j}^{(1)} X_{jj}^{(b)} X_{j1}^{(1)}
\\
        = &\tfrac{\nu_{jb}}{\nu_{1j}\nu_{11}}X_{1j}^{(1)} X_{j1}^{(1)} = \tfrac{\nu_{jb}}{\nu_{11}}X_{11}^{(1)} = \nu_{jb}I
    \end{split}
\end{equation}
Orthogonality of the $Q_b'$'s gives us
\begin{equation}\nonumber
    \begin{split}
        Q_b' Q_{b'}' & =  \sum_{jk'} V_j^{(b)}\left[ \sum_k  (V_k^{(b)})^* V_{k}^{(b')} \right] (V_{k'}^{(b')})^*   \otimes \dyad{j}{k'}
        \\
        & =  V_1^{(b)}\left[ \sum_k  (V_k^{(b)})^* V_{k}^{(b')} \right] (V_{1}^{(b')})^*   \otimes \dyad{1}{1} + ...
        \\
        & =  \sqrt{\nu_{1b}\nu_{1b'}}\left[ \sum_k  (V_k^{(b)})^* V_{k}^{(b')} \right]  \otimes \dyad{1}{1} + ...
        \\
        & = 0
    \end{split}
\end{equation}
which implies
$$
\sum_k  (V_k^{(b)})^* V_{k}^{(b')}=0, \quad \forall b\neq b'
$$
Completeness gives us
\begin{equation}\nonumber
    \begin{split}
         I_{H} = \sum_{jk} \left[ \sum_b V_j^{(b)} (V_k^{(b)})^* \right]  \otimes \dyad{j}{k} =\sum_{j} I_{H'} \otimes \dyad{j}{j}
    \end{split}
\end{equation}
and hence
$$
\sum_b V_j^{(b)} (V_k^{(b)})^* = \delta_{jk}I_{H'}
$$
\end{proof}

\begin{prop}\label{prop:rank1}
    Let $\rho: \cA_\nu \rightarrow B(H)$ ($H$ separable) be a representation, $P_a:=\rho(E_a)$ and $Q_b:= \rho(F_b)$. If the $V_j^{(b)}$'s commute, then the $P_a$'s and $Q_b$'s are direct sums of rank 1 projections.
\end{prop}

\begin{proof}
    We have
    $$
    [Y_{jk}^{(b)}, Y_{j'k'}^{(b')}]=0
    $$
    So there's an ONB $\ket{e_j}$ for $H'$ in which
    $$
    Y_{jk}^{(b)}=\sum_{\ell} (\lambda_{jk}^{(b)})_{\ell}\dyad{e_\ell}{e_\ell}
    $$
    and
    $$
    Q_b' = \sum_{jk} \sum_{\ell} (\lambda_{jk}^{(b)})_{\ell}\dyad{e_\ell}{e_\ell} \otimes \dyad{j}{k} = \sum_{\ell} \dyad{e_\ell}{e_\ell} \otimes T_{\ell}^{(b)}
    $$
    where
    $$
    T_{\ell}^{(b)}:=\sum_{jk} (\lambda_{jk}^{(b)})_{\ell} \dyad{j}{k}
    $$
    Note
    $$
    Y_{ja}^{(b)}Y_{ak}^{(b)} = \sum_{\ell} (\lambda_{ja}^{(b)})_{\ell}(\lambda_{ak}^{(b)})_{\ell}\dyad{e_\ell}{e_\ell} = \sum_{\ell} \nu_{ab} (\lambda_{jk}^{(b)})_{\ell}\dyad{e_\ell}{e_\ell}
    $$
    so in particular we have
    $$
    (\lambda_{ja}^{(b)})_{\ell}(\lambda_{ak}^{(b)})_{\ell}=\nu_{ab} (\lambda_{jk}^{(b)})_{\ell}
    $$
    Since $Q_b=Q_b^*=Q_b^2$ we must have $T_{\ell}^{(b)}=(T_{\ell}^{(b)})^*=(T_{\ell}^{(b)})^2$. We have from earlier that
    $$
    Y_{jj}^{(b)}= \tfrac{\nu_{jb}}{\nu_{11}\nu_{1j}} X_{1j}^{(1)}X_{jj}^{(b)}X_{j1}^{(1)}=\tfrac{\nu_{jb}}{\nu_{11}\nu_{1j}} X_{1j}^{(1)}  X_{j1}^{(1)}= \nu_{jb}I_{H'}
    $$
    hence we must have $(\lambda_{jj}^{(b)})_{\ell}=\nu_{jb}$ and so
    $$
    \Tr(T_{\ell}^{(b)}) = \sum_{j} (\lambda_{jj}^{(b)})_\ell = \sum_{j} \nu_{jb} = 1
    $$
    so the $T_\ell^{(b)}$'s are rank one projections. Now, notice that
    $$
    \dyad{a}{a} T_\ell^{(b)} \dyad{a}{a} = \sum_{jk} (\lambda_{jk}^{(b)})_{\ell}\braket{a}{j} \braket{k}{a} \dyad{a}{a} = (\lambda_{aa}^{(b)})_{\ell} \dyad{a}{a}= \nu_{ab} \dyad{a}{a}
    $$
    and
    \begin{equation} \nonumber
        \begin{split}
            T_{\ell}^{(b)} \dyad{a}{a} T_{\ell}^{(b)} = & \sum_{jkj'k'} (\lambda_{jk}^{(b)})_{\ell}(\lambda_{j'k'}^{(b)})_{\ell} \dyad{j}{k} \dyad{a}{a}  \dyad{j'}{k'}
            \\
            = & \sum_{jk'} (\lambda_{ja}^{(b)})_{\ell}(\lambda_{ak'}^{(b)})_{\ell}  \dyad{j}{k'}
            \\
            = & \nu_{ab} \sum_{jk'} (\lambda_{jk'}^{(b)})_{\ell}  \dyad{j}{k'} =\nu_{ab} T_{\ell}^{(b)}
        \end{split}
    \end{equation}
    We conclude that $\{\dyad{a}{a} \}_{a=1}^n$ and $\{T_\ell^{(b)} \}_{b=1}^n$ are rank 1 representations of $\cA_\nu$ for all $\ell$.
\end{proof}

\subsection{Equivalences}
Let us refer to rank 1 representations of $\cA_\nu$ as \textit{$\nu$-biased bases}.
Following \cite{MUBs} we define the following equivalence relation:
\begin{defn}[Equivalence of $\nu$-biased measurements, $\sim$]
    Let
    $$
    (\{\ket{\varphi_i^{(1)}}\}_{i=1}^n, \{\ket{\psi_j^{(1)}}\}_{j=1}^n) \text{ and } (\{\ket{\varphi_i^{(2)}}\}_{i=1}^n, \{\ket{\psi_j^{(2)}}\}_{j=1}^n)
    $$
    be two pairs of $n$-dimensional $\nu$-biased bases i.~e.~
    $$
    \abs{\braket{\varphi_i^{(k)}}{\psi_j^{(k)}}}^2= \nu_{ij}
    $$
    for $k=1,2$. Let
    $$
    M_k:=(\ket{\varphi_i^{(k)}}_j)_{ij}, \quad N_k:= (\ket{\psi_i^{(k)}}_j)_{ij}.
    $$
    Then these pairs of $\nu$-biased bases are equivalent denoted
    $$ (M_1,N_1) \sim (M_2,N_2)$$
    if and only if there exist permutation matrices $P,P'$, diagonal unitaries $D,D'$ and a unitary $U$ such that
    $$
    M_2 = U M_1 PD, \quad N_2 = U N_1 P'D'.
    $$
\end{defn}

We also need the following definition.
\begin{defn}[Equivalence of unitaries, $\approx$]
Two unitary matrices $U_1$ and $U_2$ are equivalent, written  $U_1 \approx U_2$ if and only if there exist permutation matrices $P,P'$ and diagonal unitaries $D,D'$ such that
$$
    U_2=P D U_1 D' P'.
$$
\end{defn}
This equivalence reflects the fact that permutation of basis elements as well as multiplication with phase factors leave the basis in question essentially unchanged. Now, it is clear that for any pair of $\nu$-biased bases $(M,N)$ we have
$$
(M,N)\sim (I,M^{*}N)
$$
It is also clear that since $M^{*}N$ is a basis, which is $\nu$-biased with the computational basis it must be given by a (generalized) complex Hadamard matrix
$$
M^{*}N = H = (h_{ij})_{ij}
$$
where $\abs{h_{ij}}=\sqrt{\nu_{ij}}$ and $HH^* =I_n$. Hence, a pair of $\nu$-biased bases can always be written in \textit{standard form}, $(I,H)$, for some (generalized) complex Hadamard matrix. It follows by the definitions that two pairs of $\nu$-biased bases, $(I,H)$ and $(I,H')$ are $\sim$-equivalent if and only if $H\approx H'$.

In the following, we will study in which cases $\nu$-biased bases are unique up to $\sim$-equivalence.

\begin{lem}\label{lem:uniquesol}
Let $a_1,\dots, a_n >0$ be strictly positive. If it holds $\sum_{k\neq l}a_k=a_l$ for some $1\leq l\leq n$, then the equation
\begin{align}\label{eq:triangle}
    \sum_{j=1}^{n-1} e^{2\pi i\phi_j}a_j +a_n=0
\end{align}
has a unique solution for $\phi_j\in [0,1)$, $1\leq j\leq n-1$.
\end{lem}

\begin{proof}
Assume $\sum_{k\neq l}a_k=a_l$ holds for some $l$. Multiplying \Cref{eq:triangle} with $e^{-2\pi i\phi_l}$, we obtain
\begin{align*}
    \sum_{j\neq l} e^{2\pi i(\phi_j-\phi_l)}a_j +a_l=0,
\end{align*}
where we $\phi_n=0$.
From $\sum_{k\neq l}a_k=a_l$, we get
\begin{align*}
    0=\sum_{j\neq l} e^{2\pi i(\phi_j-\phi_l)}a_j +a_l=\sum_{j\neq l} e^{2\pi i(\phi_j-\phi_l)}a_j +\sum_{j\neq l}a_j=\sum_{j\neq l} (e^{2\pi i(\phi_j-\phi_l)}+1)a_j.
\end{align*}
Since $a_j>0$, we obtain $\mathrm{Re}(e^{2\pi i(\phi_j-\phi_l)})=-1$, which yields $e^{2\pi i(\phi_j-\phi_l)}=-1$. Especially, we have $e^{-2\pi i\phi_l}=e^{2\pi i(\phi_n-\phi_l)}=-1$. This shows $\phi_l=\frac{1}{2}$ and we deduce $\phi_j=0$ for $j\neq l$.
\end{proof}

It is known for example from \cite{Haagerup}, that for $n=2,3$, all Hadamard matrices are equivalent. The next proposition gives a generalization of this result.

\begin{prop}\label{prop:Haagerup}
Consider a matrix $\nu=\{\nu_{ij}\}_{ij}$ with positive real entries such that there is a unitary matrix $H$ with $|h_{ij}|=\sqrt{\nu_{ij}}$. If for every $a,b$, there exists $l$ such that
\begin{align}\label{eq:degngon}
\sum_{k\neq l}\sqrt{\nu_{ak}}\sqrt{\nu_{bk}}=\sqrt{\nu_{al}}\sqrt{\nu_{bl}},
\end{align}
then $H'\approx H$ for every unitary $H'$ with $|h'_{ij}|=\sqrt{\nu_{ij}}$.
\end{prop}

\begin{proof}
Let $\nu$ be as in the statement and consider $H$ with $|h_{ij}|=\sqrt{\nu_{ij}}$. By multiplying with appropriate diagonal matrices $D,D'$ from left and right, we get
\begin{align*}
    DHD'=\begin{pmatrix}
        \sqrt\nu_{11}&\sqrt\nu_{12}&\dots&\sqrt\nu_{1n}\\
        \sqrt\nu_{21}&e^{2\pi i \phi_{22}}\sqrt\nu_{22}&\dots&e^{2\pi i \phi_{2n}}\sqrt\nu_{2n}\\
        \vdots&\vdots&\ddots&\vdots\\
        \sqrt\nu_{n1}&e^{2\pi i \phi_{n2}}\sqrt\nu_{n2}&\dots&e^{2\pi i \phi_{nn}}\sqrt\nu_{nn}.
    \end{pmatrix}
\end{align*}
Since $DHD'$ is a unitary matrix, we have
\begin{align*}
    \sum_{k=2}^ne^{2\pi i \phi_{ak}}\sqrt{\nu_{1k}}\sqrt{\nu_{ak}}+\sqrt{\nu_{11}}\sqrt{\nu_{a1}}=0
\end{align*}
for all $2\leq a\leq n$. Because of \Cref{eq:degngon}, we can use \Cref{lem:uniquesol} and get that all equations above have unique solutions for $\phi_{ak}\in[0,1)$, $2\leq a,k \leq n$. Therefore $DHD'$ is unique.
\end{proof}

\begin{thm}\label{thm:irrep}
Consider a matrix $\nu=\{\nu_{ij}\}_{ij}\in M_n(\C)$ for $n=2,3$ with positive real entries such that there is a unitary matrix $H$ with $|h_{ij}|=\sqrt{\nu_{ij}}$ and \Cref{eq:degngon} holds. Then the algebra $\cA_\nu$ has a unique irreducible representation up to $\sim$-equivalence, which is of dimension $n$.
\end{thm}
\begin{proof}
Due to \Cref{prop:rank1} it suffices to show that the $V_{j}^{(b)}$'s commute.
First note that since
$$
V_j^{(1)}=\sqrt{\nu_{1j}}I_{H'} \text{ and } \sum_{k}(V_k^{(b)})^*V_{k}^{(b')}=0, \ \forall b\neq b'
$$
we must have
$$
\sum_{k} V_k^{(b)}=0, \ \forall b\neq 1
$$
and since
$$
V_1^{(b)}=\sqrt{\nu_{1b}}I_{H'} \text{ and }  \sum_b V_j^{(b)} (V_k^{(b)})^* = \delta_{jk}I_{H'}
$$
we must also have
$$
\sum_b V_j^{(b)} = 0, \ \forall j\neq 1.
$$
For $n=2$ we have
$$
V_1^{(1)} = \sqrt{\nu_{11}}I_{H'}, \ V_2^{(1)} = \sqrt{\nu_{12}}I_{H'}, \ V_1^{(2)} = \sqrt{\nu_{12}}I_{H'}, \ V_2^{(2)} = -\sqrt{\nu_{12}}I_{H'}.
$$
For $n=3$ we have
\begin{equation} \nonumber
    \begin{split}
      V_1^{(1)}  = & \sqrt{\nu_{11}}I_{H'}
      \\
      V_2^{(1)} =V_1^{(2)}  = &\sqrt{\nu_{12}}I_{H'}
      \\
      V_3^{(1)}=V_1^{(3)}  = & \sqrt{\nu_{13}}I_{H'}
      \\
      V_2^{(2)} = : & V
      \\
      V_3^{(2)} = V_2^{(3)}=& -V-\sqrt{\nu_{12}}I_{H'}
      \\
      V_3^{(3)} = & V+(\sqrt{\nu_{12}}-\sqrt{\nu_{13}})I_{H'}
    \end{split}
\end{equation}
which clearly all commute.
Now, $P_a$ and $Q_b$ are direct sums of rank one $\nu$-biased measurements.
The desired now follows from Proposition~\ref{prop:Haagerup}.
\end{proof}

For $
\nu = \begin{bmatrix}
 \frac{9}{16} & \frac{1}{16} & \frac{3}{8} \\
 \frac{1}{16} & \frac{9}{16} & \frac{3}{8}  \\
 \frac{3}{8} & \frac{3}{8} & \frac{1}{4}     \end{bmatrix}
$, we further prove that $\cA_\nu$ has a unique irreducible representation up to unitary equivalence.

\begin{proof}[Proof of \Cref{lem:unique}]
    Let $\ket{\tilde{0}},\ket{\tilde{1}},\ket{\bot}$, and $\tilde{X}\tilde{Z}\tilde{X}$ be as in \Cref{sec:quantum_strategy}. For notational convenience, we let $\ket{\tilde{e}_0}:=\ket{\tilde{0}},\ket{\tilde{e}_1}:=\ket{\tilde{1}},\ket{\tilde{e}_2}:=\ket{\bot}, W:= \tilde{X}\tilde{Z}\tilde{X}$, and $\ket{\tilde{f}_i}=W\ket{\tilde{e}_i}$ for $1\leq i\leq 3$. So $(\{\ket{\tilde{e}_i}\}_{i=1}^3,\{\ket{\tilde{f}_i}\}_{i=1}^3)$ is a pair of $\nu$-biased bases. It follows that the mapping sending $E_i\mapsto \ket{\tilde{e}_i}\bra{\tilde{e}_i}$ and $F_i\mapsto \ket{\tilde{f}_i}\bra{\tilde{f}_i}$ for all $1\leq i\leq 3$ defines an irreducible representation $\tilde{\pi}$ of $\cA_\nu$ on $\C^3$.

Suppose $\pi$ is another irreducible representation of $\cA_\nu$. By \Cref{thm:irrep}, $\pi$ is $\sim$ equivalent to $\tilde{\pi}$. This means there is a pair of $\nu$-biased bases $\{\{\ket{e_i}\}_{i=1}^3, \{\ket{f_i}_{i=1}^3\}\}$ such that $\pi(E_i)=\ket{e_i}\bra{e_i}$ and $\pi(F_i)=\ket{f_i}\bra{f_i}$ for all $1\leq i\leq 3$, and that there exist permutations $\sigma,\sigma'\in S_3$ and a unitary $U$ satisfy
\begin{equation}
     \tilde{\pi}(E_i)= U \pi(E_{\sigma(i)}) U^* \text{ and }  \tilde{\pi}(F_i)= U \pi(F_{\sigma'(i)}) U^* \label{eq:unitary}
\end{equation}
   for all $1\leq i\leq$. Since both $(\{\ket{\tilde{e}_i}\}_{i=1}^3,\{\ket{\tilde{f}_i}\}_{i=1}^3)$ and $\{\{\ket{e_i}\}_{i=1}^3, \{\ket{f_i}_{i=1}^3\}\}$ are pairs of $\nu$-biased bases, by the symmetry in $\nu$, the only allowed permutations $(\sigma,\sigma')$ are $(\id,\id)$ and $((0,1),(0,1))$. If $(\sigma,\sigma')=(\id,\id)$, then \Cref{eq:unitary} implies $\pi$ and $\tilde{\pi}$ are unitarily equivalent. If $(\sigma,\sigma')=((0,1),(0,1))$, let $V:= \ket{\tilde{0}}\bra{\tilde{1}}+\ket{\tilde{1}}\bra{\tilde{0}}-\ket{\bot}\bra{\bot}$ be a unitary. We see that $\tilde{\pi}(E_i)= VU \pi(E_{i}) U^*V^* \text{ and }  \tilde{\pi}(F_i)= VU \pi(F_{i}) U^*V^*$, so we also have that  $\pi$ is unitarily equivalent to $\tilde{\pi}$. We conclude that $\tilde{\pi}$ is unique up to unitary equivalence.
\end{proof}

\end{appendix}

\bibliographystyle{alpha}
\bibliography{references}

\end{document}